\documentclass[sn-basic]{sn-jnl}

\usepackage{program}
\catcode`\|=12\relax
\usepackage{amsmath}
\usepackage{amssymb}
\usepackage{enumerate}
\usepackage{amsthm}
\usepackage{txfonts}
\usepackage{tikz-cd}
\usepackage{tikz}
\usetikzlibrary{graphs,decorations.pathmorphing,decorations.markings}
\usepackage{lineno}
\usepackage{enumitem}
\usepackage{graphics}
\usepackage{multicol}
\usepackage{pdflscape}
\usepackage{epstopdf}
\usepackage{float}
\usepackage{subfigure}
\usepackage{hyperref, cleveref}
\usepackage{natbib}

\hypersetup{breaklinks=true}

\jyear{2022}%

\theoremstyle{thmstyleone}%
\newtheorem{theorem}{Theorem}
\newtheorem{proposition}[theorem]{Proposition}%

\theoremstyle{thmstyletwo}%
\newtheorem{remark}{Remark}%

\theoremstyle{thmstylethree}%

\begin{document}

\title[The emergence of a virus variant]{The emergence of a virus variant: dynamics of a competition model with cross-immunity time-delay validated by wastewater surveillance data for COVID-19}





\author[1]{\fnm{Bruce} \sur{Pell}}

\author[2]{\fnm{Samantha} \sur{Brozak}}

\author[3]{\fnm{Tin} \sur{ Phan}}

\author[4]{\fnm{Fuqing} \sur{Wu}}

\author[2]{\fnm{Yang} \sur{Kuang}}

\affil[1]{\orgdiv{Mathematics and Computer Science}, \orgname{Lawrence Technological University}, \orgaddress{\street{21000 W. 10 Mile Rd}, \city{Southfield}, \postcode{48075}, \state{MI}, \country{USA}}}

\affil[2]{\orgdiv{School of Mathematical and Statistical Sciences}, \orgname{Arizona State University}, \orgaddress{\street{901 S. Palm Walk}, \city{Tempe}, \postcode{85287 - 1804}, \state{AZ}, \country{USA}}}

\affil[3]{\orgdiv{Theoretical Biology and Biophysics Group}, \orgname{ Los Alamos National Laboratory}, \orgaddress{ \city{Los Alamos}, \postcode{87545}, \state{NM}, \country{USA}}}

\affil[4]{\orgdiv{The University of Texas Health Science Center at Houston, School of Public Health}, \orgname{University of Texas Houston}, \orgaddress{\city{Houston}, \postcode{77030}, \state{TX}, \country{USA}}}


\abstract{ We consider the dynamics of a virus spreading through a population that produces a mutant strain with the ability to infect individuals that were infected with the established strain. Temporary cross-immunity is included using a time delay, but is found to be a harmless delay. We provide some sufficient conditions that guarantee local and global asymptotic stability of the disease-free equilibrium and the two boundary equilibria when the two strains outcompete one another. It is shown that, due to the immune evasion of the emerging strain, the reproduction number of the emerging strain must be significantly lower than that of the established strain for the local stability of the established-strain-only boundary equilibrium. To analyze the unique coexistence equilibrium we apply a quasi steady-state argument to reduce the full model to a two-dimensional one that exhibits a global asymptotically stable established-strain-only equilibrium or global asymptotically stable coexistence equilibrium. Our results indicate that the basic reproduction numbers of both strains govern the overall dynamics, but in nontrivial ways due to the inclusion of cross-immunity. The model is applied to study the emergence of the SARS-CoV-2 Delta variant in the presence of the Alpha variant using wastewater surveillance data from the Deer Island Treatment Plant in Massachusetts, USA.

}

\keywords{harmless delay, delay differential equation, COVID-19, wastewater, competitive exclusion}



\maketitle

\section{Introduction}

Viruses mutate rapidly, which may impact the clinical presentation of the disease, its epidemiology, the efficacy of therapeutics and vaccinations, or the accuracy of diagnostic tools \citep{who_variants_2022}. These mutations, along with selection pressures, may result in new variants (or strains) of a pathogen. After the emergence of SARS-CoV-2 in late 2019 \citep{who_2019}, for about 11 months, SARS-CoV-2 genomes experienced a period of relative evolutionary stasis. From late 2020, however, multiple countries began reporting the detection of SARS-CoV-2 variants that seemed to be more efficient at spreading. One of the first variants, reported on December 14, 2020 in the United Kingdom, was identified as the B.1.1.7 variant (later renamed the ``Alpha" variant). Others include the B.1.351 lineage first detected in South Africa and P.1 from four Brazilian travelers at the Haneda (Tokyo) airport \citep{who_variants_2022,niid_2021}. Since then, the World Health Organization has defined five lineages as variants of concern (Alpha, Beta, Gamma, Delta, and Omicron) \citep{who_variants_2022}. These SARS-CoV-2 variants possess sets of mutations that confer increased transmissibility and/or altered antigenicity, which the latter likely evolved in response to the immune profile of the human population having changed from naive to having been immune-imprinted from prior infections. Multiple studies have reported the rapid displacement of the Delta variant by Omicron in both clinically reported data and wastewater surveillance data \citep{lee_2022,wu_2020}. The most recent Omicron BA.4 and BA.5 lineages have also been demonstrated to resist neutralization by full-dose vaccine serum and have reduced neutralization to BA.1 infections \citep{tuekprakhon_2022}. 

COVID-19 is now one of the most widely-monitored diseases in human history, allowing for unprecedented insight into variant emergence and competition. While disease surveillance often relies on clinical case data for monitoring (and genetic sequencing to identify new variants), issues related to reporting delays or the under-reporting of cases can lead to inaccurate real-time data. Wastewater surveillance was previously used to detect poliovirus \citep{poyry_1988}, enteroviruses \citep{gantzer_1998}, and illicit drug use \citep{daughton_2001}; however, it was recently that it came to the forefront by helping fight against the COVID-19 pandemic. The rationale for SARS-CoV-2 detection in wastewater relies on the viral shedding mostly in feces and urine from infected individuals, which gives an alternative approach to recognizing viral presence and penetration in the community \citep{peccia_2020,medema_2020,ahmed_2020,fall_2022}. Quantification of viral concentrations in wastewater thus offers a complementary approach to understanding disease prevalence and predicting viral transmission by integrating with epidemiological modeling, while avoiding the same pitfalls associated with only considering clinical data.

Mathematical models have been used extensively in the study of disease dynamics with applications to the COVID-19 pandemic. Wastewater-based surveillance has increasingly been used in conjunction with mathematical and statistical models. \cite{mcmahan_2021} used an SEIR model to mechanistically relate COVID cases and wastewater data. \cite{phan2023simple} also used a standard SEIR framework, with the addition of a viral compartment, to estimate the prevalence of COVID-19 using wastewater data; results indicated that true prevalence was approximately 8.6 times higher than reported cases, consistent with previous studies (see \cite{phan2023simple} and the references therein). Naturally, the SEIR model may be extended to include heterogeneity in the viral shedding other compartments (such as those hospitalized or asymptomatic) as done by \cite{nourbakhsh_2022}.

Other studies have focused on variant emergence and competition between multiple strains or diseases. A recent study by \cite{miller_2022} used a stochastic agent-based model in an attempt to forecast the emergence of SARS-CoV-2 variants without having to previously identify a variant. The authors found that mutations are proportional to the number of transmission events and the the fitness gradient of a strain may provide insight on its persistence \citep{miller_2022}. \cite{fudolig_2020} presented a modified SIR model with vaccination (in the form of a system of ordinary differential equations) to investigate two-strain dynamics and its local stability properties. Here, individuals infected with the established strain are immediately susceptible to infection by a new strain. The authors determined that the two strains can coexist if the reproduction number of the emerging strain is lower than that of the established strain \citep{fudolig_2020}. A general multi-strain model by Arruda and colleagues \citep{arruda_2021} uses an SEIR-type model for each viral strain and uses an optimal control approach. The authors account for mitigation strategies through the inclusion of a modification terms that can reduce the contact rate of each strain, and individuals infected with a strain will have waning immunity to that same strain. However, the model does not consider cross-immunity between the strains \citep{arruda_2021}. \cite{gonzalez-parra_2021} developed a two-strain model of COVID-19 by extending the standard SEIR formulation to include asymptomatic transmission and hospitalization. The study found that the introduction of a slightly more transmissible strain can become dominant in the population \citep{gonzalez-parra_2021}. These models may also include a time delay to account for various biological phenomena. For example, \cite{rihan_2020} developed a delayed stochastic SIR model with cross-immunity, where a time delay was incorporated to adjust for the incubation period of a disease and stochasticity was used to determine the effect of randomness on parameters.

In this paper, we present a four-dimensional modified SIR model to study disease dynamics when two strains are circulating in a population. A time delay is incorporated to account for temporary cross-immunity induced by infection with an established (or dominant) strain.  This paper is organized as follows: in \autoref{sec:model}, the model is formulated and the equilibria of the full system are analyzed. Interestingly, we find that the time delay does not influence the stability of equilibria and is hence a harmless delay \citep{Gopalsamy_1983,Driver_72}. In \autoref{sec:tran_model}, we introduce the transient model to study global stability of the coexistence equilibrium, and bifurcation curves are shown. Finally, the model is calibrated using wastewater data and the results are studied using a sensitivity analysis in \autoref{sec:num}.

\section{The general model}\label{sec:model}

In this section, we introduced our mathematical model that incorporates two competing virus strains and conduct basic model analysis. 

We consider a population-level virus competition model using a compartmental framework. We let $S(t)$, $I_1(t)$, $I_2(t)$ and $R_l(t)$ be the individuals that are susceptible to both virus strains, infectious with strain 1, infectious with strain 2 and recovered from strain 1 but susceptible to strain 2 at time $t$, respectively. Let $\tau$ be the time it takes for an individual infected with strain 1 to become susceptible to infection by strain 2. We introduce the following \textit{two-strain virus competition model with temporary cross-immunity}:
\begin{equation}
\begin{aligned}
\frac{dS}{dt} & =  a - dS -\beta_1 S I_1 - \beta_2 S I_2,  & \\
\frac{dI_1}{dt} & = \beta_1SI_1 -\gamma_1 I_1 -d_1I_1, & \\
\frac{dI_2}{dt} & =  \beta_2 S I_2 +\beta_2 R_l I_2 -\gamma_2 I_2-d_2I_2,  &\\
\frac{dR_l}{dt} & = \gamma_1 I_1(t-\tau) - \beta_2 R_l I_2 -dR_l.
\label{eqn: model1}
\end{aligned}
\end{equation}
An ODE version of this model without demography, independently developed, was used recently to describe the evolutionary dynamics of SARS-CoV-2 on the population level \cite{boyle2022selective}. As a practical convention, all parameters in our model are positive. The birth rate of susceptible individuals is constant at rate $a$. Susceptible individuals die naturally at rate $d$. Infected individuals with strain 1 or strain 2 die at rate $d_1$ or $d_2$, respectively. To investigate how disease-induced death influences virus strain competition, we make the distinction that $d_1$ and $d_2$ are disease-induced death rates, while $d$ is the natural death rate. In practice, $d_1\geq d$ and $d_2\geq d$. In system \eqref{eqn: model1}, susceptible individuals become infectious when they come into contact with infectious individuals from either strain at rates $\beta_1$ and $\beta_2$, respectively. Individuals infected with strain 1 recover at rate $\gamma_1$ and enter the $R_l$ compartment where they are immune to strain 1, but become susceptible to strain 2 at rate $\beta_2$ after $\tau$ days has passed. Infectious individuals with strain 2, recover at rate $\gamma_2$. We note that
\begin{equation}
R_l(t) = \int_0^{t-\tau} \gamma_1 I_1(u) e^{-\int_{u+\tau}^{t} \beta_2 I_2(\sigma) + d \ d\sigma} \ du
\label{eqn: R_l integral}
\end{equation}
and differentiating with respect to $t$ we have

\begin{equation*}
    \frac{dR_l}{dt} = \gamma_1 I_1(t-\tau) - \beta_2 R_l I_2 -dR_l.
\end{equation*}  

We assume that the transition rates from $S$ to $I_1$, $S$ to $I_2$ and $R_l$ to $I_2$ follow the classical mass action law and all other transition rates are proportional to the compartment being left or entered. \autoref{fig: model diagram} shows a summarizing schematic of the model transitions. We note that using standard incidence for the disease transmission rates would make more biological sense, since it shouldn't matter how many people have the disease around you, only how many you come into contact with. Lastly, initial histories for system \eqref{eqn: model1} are prescribed by:
$$(S(u),I_1(u),I_2(u),R_l(u))=(\phi_1(u),\phi_2(u),\phi_3(u),\phi_4(u))$$
where $\phi_i(u)$, $i=1,2,3,4$ are bounded, continuous and nonnegative functions for $u\in [-\tau,0).$

\begin{figure}
    \centering
    \includegraphics[width=\textwidth]{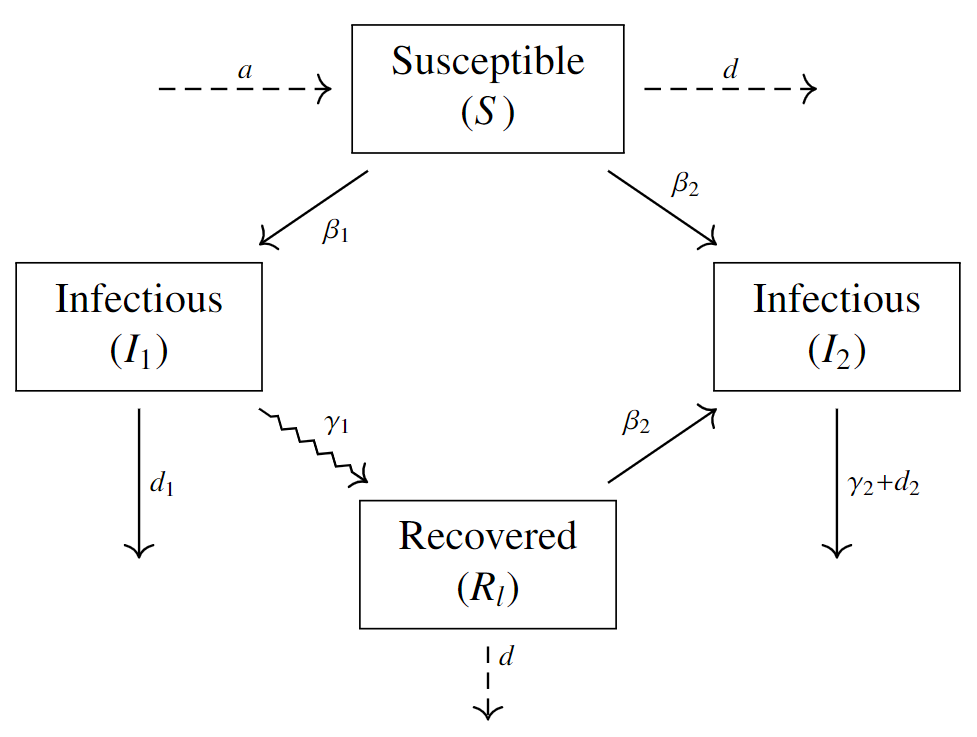}
    \caption{Schematic of the general model, system \eqref{eqn: model1}. Solid arrows correspond to disease-related transitions, dashed lines correspond to demographic transitions (birth and natural death) and the squiggle arrow corresponds to the time delay for an individual that has recovered from strain 1 to be susceptible to strain 2.}
    \label{fig: model diagram}
\end{figure}

\subsection{Non-negativity and boundedness}

We notice that the vector-valued function (1) and its derivative exist and are continuous. Therefore there exists a unique noncontinuable solution defined on some interval $[-\tau,s)$ where $s>0$ \citep{kuang_1993,Smith2011}. Our first step is to show that the model produces solutions that are biologically plausible. We prove this with the following two propositions. We first show that if solutions start nonnegative, then they will stay nonnegative on $[0,s)$. After that, we show solutions remain bounded for all time, which then implies $s=+\infty$ by Theorem 3.2 and Remark 3.3 in \citep{Smith2011}.

\begin{proposition}
Solutions to system (1) that start nonnegative stay nonnegative.
\label{th: full positivity}
\end{proposition}
\begin{proof}
Observe that if $I_1(0)=0$, then $I_1(t)=0$ for $t\in[0,s)$. Similarly, $I_2(t)=0$ for $t\in[0,s)$ if $I_2(0)=0$. Thus we may assume that $I_1(0)>0$ and $I_2(0)>0$. Let $y(t) = d+\beta_1I_1+\beta_2I_2$ and observe that the first equation in system (1) can be rewritten as
\begin{equation*}
    S'(t) = a -y(t)S.
\end{equation*}
Applying the integrating factor method we obtain $$S(t) = \left(\int_0^t ae^{\int_0^\zeta y(\xi) d\xi} d\zeta  +S(0)\right)e^{-\int_0^t y(\xi) d\xi}> 0.$$ This implies that $S(t)$ is positive for $t\in[0,s)$. From the second equation of system (1) for $t\in[0,s)$ we have\begin{equation*}
    I_1'(t)\geq -\left(\gamma_1+d_1\right)I_1. 
\end{equation*} This implies $$I_1(t)\geq I_1(0)e^{-(\gamma_1+d_1)t}>0.$$ We see that $$I_2'(t)=p(t)I_2,$$ where $p(t) = \beta_1S+\beta_2R_l-\gamma_2-d_2.$ This implies that $$I_2(t)=I_2(0)e^{\int_0^t p(\xi) d\xi}>0$$ for $t\in[0,s)$. From equation (2) we have that $$R_l(t) = \int_0^{t-\tau} \gamma_1 I_1(u) e^{-\int_{u+\tau}^{t} \beta_2 I_2(\sigma) + d \ d\sigma} \ du>0$$ for $t\in[0,s)$. Hence, solutions with nonnegative initial conditions will remain nonnegative.
\end{proof}

 Throughout the rest of this paper, we assume that $S(0)>0,\,I_1(0)>0,\, I_2(0) \geq 0,\,R_l(0)\geq 0.$ \, $N=S+I_1+I_2+R_l$, and $N(0)=a/d=S(0)+I_1(0)+I_2(0)+R_l(0).$

\begin{proposition}\label{th: full bounded} Solutions to system \eqref{eqn: model1} are bounded.
\begin{proof}
Let $N_1=S+I_1.$ 
Since components of solutions are nonnegative and $I_1(t)>0$ for $t\in[-\tau,s)$, we have
    \begin{align*}
        N_1'(t) & < a-dS -d_1I_1\\
        &\leq a -\alpha_1 N_1,
    \end{align*}
where $\alpha_1 := \min\left\{d,d_1\right\}$. Hence
$$ S(t)+I_1(t)<a/\alpha_1.$$ In particular, we see that $I_1(t)\leq\max\left\{I_1(0),\frac{a}{\alpha_1}\right\}:=B$ for $t\in [0,s).$
Define
    $$f(t) = N(t) + \gamma_1 \int_{t-\tau}^t I_1(s) \ ds.$$
Then
    \begin{align*}
        f'(t) &= a-dS -d_1I_1 - (\gamma_2 + d_2) I_2 - dR_l\\
        &\leq a -\alpha N = a + \alpha \gamma_1 \int_{t-\tau}^t I_1(s) \ ds-\alpha f \leq a + \alpha \gamma_1\tau B -\alpha f,
    \end{align*}
where $\alpha := \min\left\{d,d_1,\gamma_2+d_2\right\}$. Let $A=a + \alpha \gamma_1\tau B $. This yields $f(t)\leq \frac{A}{\alpha}+\left(N(0)-\frac{A}{\alpha}\right)e^{-\alpha t}$, hence $f(t)\leq \max\{N(0),\frac{A}{\alpha}\}.$ Since $N(t) \leq f(t),$ this proves boundedness of solutions. 
\end{proof} 
\end{proposition}
The fact that solutions are bounded for all $t\geq0$ implies that $s=+\infty$.
\subsection{Analysis of equilibria}


In order to gain a global understanding of the dynamics of system \eqref{eqn: model1}, we study the existence, number and stability of its equilibria. For infectious disease models, the dynamics can usually be characterized using the basic reproduction number \citep{delamater_2019}. By the next generation matrix method in \citep{VANDENDRIESSCHE200229} we find the basic reproduction numbers for strain 1 and strain 2 to be

\begin{equation*}
    R_{i} = \frac{a\beta_i}{d(d_i+\gamma_i)}
\end{equation*} for $i=1, 2$.
The full system exhibits four biologically relevant steady states: a disease-free steady state ($E_0$), two steady states where either strain 1 outcompetes strain 2 ($E_1$) or strain 2 outcompetes strain 1 ($E_2$), and a coexistence steady state ($E_c$). They take the following forms:

\begin{equation}
\label{eqn: E_0}
    E_{0} = \left(\frac{a}{d}, 0,0,0\right),
\end{equation}   

\begin{equation}
\label{eqn: E_1}
    E_1 = \left(\frac{\gamma_1+d_1}{\beta_1}, \frac{d}{\beta_1}\left(R_{1} - 1\right),0,\frac{\gamma_1}{\beta_1}\left(R_{1}-1\right)\right),
\end{equation}   

\begin{equation}
\label{eqn: E_2}
    E_2 = \left(\frac{\gamma_2+d_2}{\beta_2}, 0,\frac{d}{\beta_2}\left(R_{2} - 1\right),0\right),
\end{equation}   

\begin{equation}
\label{eqn: E_c}
  E_c = \left(\frac{\gamma_1+d_1}{\beta_1}, I_1^*,I_2^*,\frac{a}{d}\left(\frac{R_1-R_2}{R_{1}R_{2}}\right)\right),  
\end{equation}where$$I_1^*=\frac{d\left(R_{1}-R_{2}\right)}{d_1\left(\frac{a\beta_1}{dd_1} - R_{2}\right)}$$and$$I_2^*=\frac{d\left(\gamma_1R_{1}+d_1\right)\bigg[R_{2}-\frac{a\beta_1}{d\left(\gamma_1R_{1}+d_1\right)}\bigg]}{d_1\beta_2\left(\frac{a\beta_1}{d_1d}-R_{2}\right)}$$

We see that $E_0$ always exists, $E_1$ exists when $R_{1}>1$, and $E_2$ exists when $R_{2}>1$. From the form of $E_c$ we see that it exists when $\frac{a\beta_1}{d\left(\gamma_1R_{1}+d_1\right)}<R_{2}<R_{1}$ and $R_2<\frac{a\beta_1}{dd_1}.$ Notice that since $R_1<\frac{a\beta_1}{dd_1}$ is always true, we can equivalently say $\frac{a\beta_1}{d\left(\gamma_1R_{1}+d_1\right)}<R_{2}<R_{1}<\frac{a\beta_1}{dd_1}$.
We summarize the above discussion in the following proposition.
\begin{proposition}The following are true for system \eqref{eqn: model1}. 
\label{eq exist}
\begin{enumerate}
\item The disease-free equilibrium, $E_0$, always exists.
\item The boundary equilibria, $E_i$, exist when $R_i>1$ for $i=1, 2.$
    \item The coexistence equilibrium, $E_c$, exists exactly when $\frac{a\beta_1}{d\left(\gamma_1R_{1}+d_1\right)}<R_{2}   <R_{1}$ and $R_2<\frac{a\beta_1}{dd_1}$.
\end{enumerate}
\end{proposition}

\begin{remark}
If we let 
\begin{equation}
    f(R_1)=\frac{a\beta_1}{d\left(\gamma_1R_{1}+d_1\right)},
    \label{eqn: R2 curve}
\end{equation} then we see that $f\left(1\right)=\frac{a\beta_1}{d\left(d_1+\gamma_1\right)}=R_1.$ That is, if $R_1=1$ then $R_2=1$ on this curve.
\end{remark}
We have the following result for $E_0$.
\begin{proposition}
\label{LAS E0}
$E_0$ is locally asymptotically stable when $\max\{R_1,R_2\}<1$. $E_0$ is unstable when $R_{1}>1$ or $R_{2}>1$.\begin{proof}
The Jacobian matrix evaluated at $E_0$ is 
\begin{equation*}
\begin{pmatrix}
-d& -\frac{a\beta_1}{d} & -\frac{a\beta_2}{d} & 0\\
0 & \left(d_1+\gamma_1\right)\left(R_{1} -1\right) & 0& 0\\
0&0 &\left(d_2+\gamma_2\right)\left(R_{2} - 1\right) & 0\\
0&\gamma_1e^{-\lambda\tau} &0 & -d
\end{pmatrix}.
\end{equation*}

The corresponding eigenvalues are

\begin{equation*}
    \lambda_{1,2} = -d,
\end{equation*}

\begin{equation*}
    \lambda_3 = \left(d_1+\gamma_1\right)\left(R_{1} -1\right),
\end{equation*}

\begin{equation*}
    \lambda_4 = \left(d_2+\gamma_2\right)\left(R_{2} - 1\right).
\end{equation*}
Therefore, $E_0$ is locally asymptotically stable whenever $\max\{R_1,R_2\}<1$. It is unstable whenever either $R_{1}>1$ or $R_{2}>1$.
\end{proof}
\end{proposition}

In addition to local stability, we have the following global stability result for $E_0$.
\begin{theorem}\label{GAS E0}
If $\max\{R_1,R_2\}<1$, then the disease-free equilibrium $E_0$ is globally asymptotically stable.
 
\begin{proof}
Observe that
    \begin{align*}
        \frac{dS}{dt}&= a - dS - \beta_1SI_1 - \beta_2SI_2\\
        &\leq a - dS,
    \end{align*}
which implies $\limsup_{t\to\infty}S(t)      \leq a/d$. If $S(0)<a/d$, then $S(t)\leq a/d$ for all $t>0$. If $S(0)>a/d$, then $\frac{dS}{dt}<0$. Hence, the region $\{S(t): 0\leq S(t)\leq a/d \}$ is positively invariant and attracting.

By assumption, $R_1 = \frac{a\beta_1}{d(\gamma_1+d_1)}<1$, which implies that there exists $\varepsilon_1>0$ such that $\beta_1(a/d + \varepsilon_1) < \gamma_1 + d_1$. For this $\varepsilon_1$, there exists $t_1>0$ such that, for $t>t_1$, $S(t)< a/d + \varepsilon_1$. Then, for $t>t_1$,
    \begin{align*}
        \frac{dI_1}{dt} &= \beta_1SI_1 - (\gamma_1 + d_1)I_1\\
        &<\left[\beta_1\left(\frac{a}{d} + \varepsilon_1\right) - (\gamma_1 + d_1) \right]I_1.\\
    \end{align*}
Since $\beta_1(a/d + \varepsilon_1) < \gamma_1 + d_1$, $I_1(t)$ is exponentially decreasing for $t>t_1$. This result implies that $\liminf_{t\to\infty}I_1(t)\leq0$. Since $I_1(t)\geq0$ for all $t>0$, $\limsup_{t\to\infty}I_1(t) \geq 0$. Hence $\lim_{t\to\infty}I_1(t)=0$. 

Since $\lim_{t\rightarrow \infty}I_1(t) = 0$, we see that for any $\varepsilon>0,$ there is a $t_*= t_*(\varepsilon)$ such that for $t>t_*$, we have
    $$\frac{dR_l}{dt} = \varepsilon-\beta_2I_2R_l - dR_l<\varepsilon - dR_l $$
and a similar argument can be used to show that $R_l$ is eventually bounded by $\varepsilon/d$ and hence $R_l\to0$ as $\varepsilon \to 0$ and $ t \to \infty.$

Finally, the assumption that $R_2 = \frac{a\beta_2}{d(\gamma_2+d_2)}<1$ implies that there exists $\varepsilon_2>0$ such that $\beta_2 \left(a/d + \varepsilon_1 + \varepsilon_2\right)<\gamma_2+d_2$. Since $R_l\to0$ as $t\to\infty$, there exists $t_2>0$ such that, for this $\varepsilon_2$, $R_l<\varepsilon_2$ for $t>t_2$. Then, for $t>t_2$,
    \begin{align*}
        \frac{dI_2}{dt} &= \beta_2SI_2 + \beta_2R_lI_2 - (\gamma_2 + d_2)I_2\\
        &< \beta_2 \left(\frac{a}{d} + \varepsilon_1\right)I_2 + \beta_2(\varepsilon_2)I_2 - (\gamma_2 + d_2)I_2\\
        &= \left[\beta_2\left(\frac{a}{d}+\varepsilon_1 + \varepsilon_2 \right) - (\gamma_2 + d_2)\right]I_2
    \end{align*}
Thus $I_2(t)$ is exponentially decreasing. By a similar argument as above, $\lim_{t\to\infty}I_2(t) = 0$.

We have shown that $\lim_{t\to\infty}I_1(t) = 0$, $\lim_{t\to\infty}I_2(t) = 0$, and $\lim_{t\to\infty}R_l(t)=0$. Thus we obtain the limiting equation
        $$\frac{dS}{dt} = a-dS$$
which implies that $\lim_{t\to\infty}S(t) = a/d$. This concludes the proof.
\end{proof}
\end{theorem}
The conditions that govern the global stability of $E_0$ make good biological sense. 

We have the following result for $E_1$.
\begin{proposition}
\label{LAS E1}
If $R_{1}>1$, then $E_{1}$ exists. Furthermore we have,
\begin{enumerate}
    \item If $R_{2}<\frac{a\beta_1}{d\left(\gamma_1R_{1}+d_1\right)}$, then $E_1$ is asymptotically stable.
    \item If $R_{2}>\frac{a\beta_1}{d\left(\gamma_1R_{1}+d_1\right)}$, then $E_1$ is unstable.
\end{enumerate}

\begin{proof}
The Jacobian matrix evaluated at $E_1$ is 
\begin{equation*}
\begin{pmatrix}
-dR_{1}& -\left(d_1+\gamma_1\right) & -\frac{\beta_2\left(d_1+\gamma_1\right)}{\beta_1} & 0\\
d\left(R_{1}-1\right) & 0 & 0& 0\\
0&0 & \frac{\beta_2 d_1}{\beta_1} -(d_2+\gamma_2) + \frac{\gamma_1R_{1}\beta_2}{\beta_1} & 0 \\
0&\gamma_1e^{-\lambda\tau} & -\frac{\beta_2\gamma_1}{\beta_1}\left(R_{1}-1\right) & -d
\end{pmatrix}.
\end{equation*}
The corresponding characteristic polynomial factors to

\begin{equation*}
    h(\lambda) = g(\lambda)\left(-d-\lambda\right)\left(\frac{\left(d_2+\gamma_2\right)\left(d_1d+\gamma_1dR_{1}\right)}{a\beta_1}\left(R_{2}-\frac{a\beta_1}{d_1d+\gamma_1dR_{1}}\right)-\lambda\right),
\end{equation*}
where
    \begin{equation*}
        g(\lambda) = d\left(d_1+\gamma_1\right)\left(R_{1}-1\right) + dR_{1}\lambda +\lambda^2.
    \end{equation*}

Therefore, the corresponding roots are
\begin{equation*}
\begin{aligned}
    \lambda_1 = & -d,\\
    \lambda_2 = & \frac{\left(d_2+\gamma_2\right)\left(d_1d+\gamma_1dR_{1}\right)}{a\beta_1}\left(R_{2}-\frac{a\beta_1}{d_1d+\gamma_1dR_{1}}\right),
\end{aligned}
\end{equation*}
and the roots to the quadratic equation 
$g(\lambda) = d\left(d_1+\gamma_1\right)\left(R_{1}-1\right) + dR_{1}\lambda +\lambda^2.$
Consequently, $\lambda_1<0$ and $\lambda_2<0$ by assumption (1). In addition, by the Routh-Hurwitz stability criterion for quadratic equations \citep{brauer_2012}, $g(\lambda)$ has roots with negative real parts since $R_{1}>1$. Thus all eigenvalues have negative real part.
Lastly, we see that $E_1$ is unstable if and only if $\lambda_2>0$, that is, $R_{2}>\frac{a\beta_1}{d\left(d_1+\gamma_1R_{1}\right)}$. This concludes the proof.
\end{proof}
\end{proposition}

This previous proposition shows that due to the immune evasion of the emerging strain, the reproduction number of the emerging strain must be significantly lower than that of the established strain for it to competitively exclude the emerging strain. In addition to local asymptotic stability, we have the following result for global stability of $E_1$.
\begin{theorem}
\label{GS E1}
If $R_1>1$, $R_2<1$, $R_{2}<\frac{a\beta_1}{d_1d+\gamma_1dR_{1}}$ and $\beta_2\left(\frac{a}{d} + D\right) < \gamma_2 + d_2$ where $D=\frac{\gamma_1 C}{d}$, $C=\frac{\beta_1 B}{\gamma_1+d_1}$ and $B=\left(\frac{a}{2d}\right)^2$, then $E_1$ is globally asymptotically stable.

\begin{proof}
    Consider the sum $S(t) + I_1(t)$. Observe that
        \begin{align*}
            \frac{dS}{dt} + \frac{dI_1}{dt} &= a - dS - \beta_2SI_2 - (\gamma_1 + d_1)I_1\\
            &\leq a - dS - (\gamma_1 + d_1)I_1\\
            &\leq a - \alpha(S+ I_1) \tag{$*$}\label{res1}
        \end{align*}
    where $\alpha = \min\{d,\gamma_1 + d_1\}$. In practice, we assume that $d_1\geq d$ and so $\alpha = d$. Recall that $S(0)+I_1(0)\leq a/d$. We see that $S(t) + I_1(t) \leq a/d$; hence both $S(t)$ and $I_1(t)$ are bounded above. Furthermore, because the arithmetic mean is greater than or equal to the geometric mean, $SI_1 \leq \left(\frac{1}{2} \frac{a}{d}\right)^2 = B$. Hence we obtain,
    
        $$\frac{dI_1}{dt}\leq \beta_1 B - \left(\gamma_1 +d_1\right)I_1$$
        and therefore, $\limsup_{t \rightarrow \infty}{I_1}\leq \frac{\beta_1B}{\gamma_1+d_1}=C.$ 
        
        Since $\beta_2\left(\frac{a}{d} + D\right) < \gamma_2 + d_2,$ we see that there is small constant $\varepsilon_0>0$ such that 
        $$\beta_2\left(\frac{a}{d} + D +3\varepsilon_0\right) < \gamma_2 + d_2.
        $$Let $\varepsilon>0$ and $\varepsilon_0=\frac{\gamma_1}{d}\varepsilon$. Thus there exists $t_{\varepsilon}>0$ such that $\limsup_{t\rightarrow\infty}{I_1}< \frac{\beta_1 B}{\gamma_1+d_1}+\varepsilon.$ for $t>t_{\varepsilon}$. Therefore, for $t>t_{\varepsilon}$ we have 
        
        $$\frac{dR_l}{dt}< \gamma_1\left(\frac{\beta_1 B}{\gamma_1+d_1}+\varepsilon\right) - dR_l$$
        and we obtain $$\limsup_{t\rightarrow\infty}{R_l}\leq\frac{\gamma_1}{d}\left(\frac{\beta_1 B}{\gamma_1+d_1}+\varepsilon\right)=D+\frac{\gamma_1}{d}\varepsilon.$$ Therefore, for any $\varepsilon_0$ there exists $t_0>t_\varepsilon$ such that $$\limsup_{t\rightarrow\infty}{R_l}\leq\frac{\gamma_1}{d}\left(\frac{\beta_1 B}{\gamma_1+d_1}+\varepsilon\right)=D+\varepsilon_0.$$ Hence there exists $t_2>$ such that for $t>t_2$, $R_l<D+2\varepsilon_0$. In addition, since $R_2<1$, there exists $\varepsilon_1>0$ such that $\beta_2(a/d + \varepsilon_1) < \gamma_2 + d_2$. Our previous result \eqref{res1} implies that, for this $\varepsilon_1$, there exists $t_1>0$ such that $S(t) < a/d + \varepsilon_1$ for $t>t_1$. 
        We are now ready to control $I_2$. We have the following,
        
        \begin{align*}
          \frac{dI_2}{dt} & = \beta_2 SI_2 +\beta_2 R_l I_2 -\left(\gamma_2 +d_2\right)I_2\\
             & < \beta_2 \left(\frac{a}{d}+\varepsilon_1\right) I_2 + \beta_2R_lI_2-\left(\gamma_2 +d_2\right)I_2\\
             &<\left[\beta_2\left(\frac{a}{d}+\varepsilon_1 \right)+ \beta_2\left(D+2\varepsilon_0\right) - \left(\gamma_2 +d_2\right)\right]I_2\\
             &= \left[\beta_2\left(\frac{a}{d}+\varepsilon_1 + D+2\varepsilon_0\right) - \left(\gamma_2 +d_2\right)\right]I_2
        \end{align*}
    
    Letting $\varepsilon_1=\varepsilon_0$ we obtain    
$$\frac{dI_2}{dt} < \left[\beta_2\left(\frac{a}{d}+3\varepsilon_0 + D\right) - \left(\gamma_2 +d_2\right)\right]I_2$$
        
    Thus, for $t>t_1$, $I_2(t)$ is exponentially decreasing, implying that $\liminf_{t\to\infty}I_2(t)\leq0$. However, since $I_2(t)$ is non-negative, $\limsup_{t\to\infty}I_2(t)\geq0$. Hence $\lim_{t\to\infty}I_2(t) = 0$.
    
    Observe that once $I_2$ goes to zero, $R_l$ does not impact the dynamics of the model, allowing us to consider the behavior of the resulting two-dimensional system:
        \begin{align*}
            \frac{dS}{dt} &= a - dS - \beta_1SI_1,\\
            \frac{dI_1}{dt} &= \beta_1SI_1 - \gamma_1 I_1 - d_1 I_1.
        \end{align*}
        It's easy to see that this system has a positive equilibrium point, $E^* = \left(\frac{\gamma_1+d_1}{\beta_1}, \frac{d}{\beta_1}\left(R_{1} - 1\right)\right)$ which is globally asymptotically stable. Finally, considering the limiting profile of $R_l$ we obtain, $\lim_{t \rightarrow \infty}R_l(t) =\frac{\gamma_1}{\beta_1}\left(R_{1}-1\right).$ Therefore, all trajectories of system \eqref{eqn: model1} tend to $E_1$.
\end{proof}


We have the following result for $E_2$.
\begin{proposition}
\label{LAS E2}
If $R_{2}>1$, then $E_{2}$ exists. Furthermore we have,
\begin{enumerate}
    \item If $\frac{R_{1}}{R_{2}}<1$, then $E_2$ is asymptotically stable.
    \item If $\frac{R_{1}}{R_{2}}>1$, then $E_2$ is unstable.
\end{enumerate}

\begin{proof}
The Jacobian matrix evaluated at $E_2$ is 
\begin{equation*}
\begin{pmatrix}
 -dR_{2}& -\frac{\beta_1}{\beta_2}\left(d_2+\gamma_2\right) & -\left(d_2+\gamma_2\right)& 0\\
0& \left(d_1+\gamma_1\right)\left(\frac{R_{1}}{R_{2}}-1\right)  & 0& 0\\
d\left(R_{2}-1\right)&0 & 0 & d\left(R_{2}-1\right) \\
0&\gamma_1e^{-\lambda\tau} & 0 & -dR_{2}
\end{pmatrix}.
\end{equation*}
The corresponding characteristic polynomial factors to

\begin{equation*}
   h(\lambda) = \left(\left(d_1+\gamma_1\right)\left(\frac{R_{1}}{R_{2}}-1\right)-\lambda\right)\left(-dR_{2}-\lambda\right)\left(d\left(d_2+\gamma_2\right)\left(R_{2}-1\right) + dR_{2}\lambda +\lambda^2\right).
\end{equation*}
Therefore, the corresponding roots are
\begin{equation*}
\begin{aligned}
    \lambda_1 = & \left(d_1+\gamma_1\right)\left(\frac{R_{1}}{R_{2}}-1\right),\\
    \lambda_2 = & -dR_{2},
\end{aligned}
\end{equation*}
and the roots to the quadratic equation 
$g(\lambda) = d\left(d_2+\gamma_2\right)\left(R_{2}-1\right) + dR_{2}\lambda +\lambda^2.$
Consequently, $\lambda_1<0$ by assumption (1) and $\lambda_2<0$. In addition, by the Routh-Hurwitz stability criterion for quadratic equations \citep{brauer_2012}, $g(\lambda)$ has roots with negative real parts since $R_{2}>1$. Thus all eigenvalues have negative real part.
Lastly, we see that $E_2$ is unstable if and only if $\lambda_{1}>0$, that is, $\frac{R_{1}}{R_{2}}<1$. This concludes the proof.

\end{proof}
\end{proposition}
\begin{remark}
From the above local stability results we see that for $\tau\geq0$, the equilibria $E_0$, $E_1$, and $E_2$ do not undergo a delay-induced stability switch. For this reason, we call the delay, $\tau$, a harmless delay \citep{Gopalsamy_1983,Driver_72}. We summarize this formally in the next theorem.
\end{remark}
\begin{proposition}
For $\tau\geq 0$, the equilibria $E_0$, $E_1$, and $E_2$ do not undergo a delay-induced stability switch.
\begin{proof}
We see that by Propositions \ref{LAS E0}, \ref{LAS E1}, and \ref{LAS E2}, $\tau$ does not appear in any of the characteristic polynomials and therefore does not influence stability.
\end{proof}
\end{proposition}

\begin{theorem}
\label{thm: GS E2}If $R_1<1$ and $R_2>1$, then $E_2$ is globally asymptotically stable.
\end{theorem}
\begin{proof}
  Since $R_1<1$, $E_1$ and $E_c$ do not exist. By assumption $\beta_1\frac{a}{d}<\gamma_1+d_1$, thus there exists $\varepsilon>0$ such that $\beta_1\left(\frac{a}{d}+\varepsilon\right)<\gamma_1+d_1$. Since $S\leq \frac{a}{d}+\left(S(0)-\frac{a}{d}\right)e^{-dt}$ there exists $t_\varepsilon>0$ such that for $t>t_{\varepsilon}$ we have $S<\frac{a}{d}+\varepsilon$. Hence,

\begin{equation}
\begin{aligned}
I_1'(t)&=\beta_1SI_1-\left(\gamma_1+d_1\right)I_1\\
     & < \left[\beta_1\left(\frac{a}{d}+\varepsilon\right)-\left(\gamma_1+d_1\right)\right]I_1\\
\end{aligned}
\end{equation}
and therefore $I_1(t)<I_1(0)e^{(\beta_1\left(\frac{a}{d}+\varepsilon\right)-(\gamma_1+d_1)})t$, and so $I_1(t)\rightarrow 0$ as $t\rightarrow\infty$. Since $I_1\rightarrow 0$, for $\varepsilon_1>0$ there exists $t_1$ such that $\gamma_1I_1(t-\tau)<\varepsilon_1$ for $t>t_1.$ Therefore,
\begin{equation}
\begin{aligned}
R_l'(t)&=\gamma_1I_1(t-\tau) -\beta_1I_2R_l-dR_l\\
     & < \varepsilon_1 -dR_l\\
\end{aligned}
\end{equation}
which implies that $\limsup_{t\to\infty}R_l(t)\leq \frac{\varepsilon_1}{d}.$ Letting $\varepsilon_1\to0$ we obtain $\limsup_{t\to\infty}R_l(t)\leq 0.$ In addition, since $R_l\geq0$ we have $\liminf_{t \to\infty} R_l\geq0$. Therefore, $\lim_{t\to\infty}R_l(t)=0$ and we obtain the 2 dimensional limiting system: 
\begin{equation}
\begin{aligned}
S' & =  a - dS -\beta_2SI_2\\
I_2' & = \beta_2SI_2 - \left(\gamma_2 + d_2\right)I_2.
\end{aligned}
\label{eqn: limiting system}
\end{equation}
Let $S^*=\frac{\gamma_2+d_2}{\beta_2}$ and $I_2^*=\frac{a}{\beta_2}\left(r_{2} - \frac{d}{a}\right)$ and consider the following Lyapunov function 
\begin{equation}
    V(S,I_2)=S-S^*\ln(S)+I_2-I_2^*\ln(I_2).
\end{equation}
Then the derivative with respect to time is given by

\begin{align*}
    \dot{V}&= \left(1-\frac{S^*}{S}\right)\left(a-dS-\beta_2SI_2\right) + \left(1-\frac{I_2^*}{I_2}\right)\left(\beta_2S-\left(\gamma_2+d_2\right)\right)I_2 \\
     &=-\frac{1}{S}\left(S-S^*\right)\left(d(S-S^*) +\beta_2S(I_2-I_2^*) +\beta_2 I^*_2(S-S^*)\right) + \beta_2\left(I_2 - I_2^*\right)\left(S-S^*\right)\\
     &=-\frac{d}{S}\left(S-S^*\right)^2-\frac{\beta_2 I_2^*}{S}\left(S-S^*\right)^2\\
     &\leq0.
\end{align*}

We have used the steady state relationships $\gamma_2+d_2=\beta_2 S^*$ and $a=dS^*+\beta_2 S^*I_2^*$. Thus we have a Lyapunov function. We have that $E=\{(S,I_2)|\dot{V}(S,I_2)=0\}=\{(S^*,I_2)| I_2>0 \}$. Let the largest invariant set of $E$ be $M$. Since $S(t)=S^*,$ we have $S'(t)=0=a-dS^*-\beta_2 I_2S^*$ which implies that $I_2=I_2^*$.
Hence the largest invariant set of $E$ is $$M=\{ (S^*, I_2^*) \}=\Biggl\{\left(\frac{\gamma_2+d_2}{\beta_2},\frac{d}{\beta_2}\left(R_{2} - 1\right)\right)\Biggr\}.$$
Thus all solutions of system \eqref{eqn: limiting system} tend to $\left(\frac{\gamma_2+d_2}{\beta_2},\frac{d}{\beta_2}\left(R_{2} - 1\right)\right)$, by the Lyapunov-LaSalle Theorem. This shows that all solutions to system \eqref{eqn: model1} tend to $E_2$ if $R_1<1$ and $R_2>1$.
\end{proof}
\end{theorem}

\begin{figure}[!htp]
\includegraphics[width=\textwidth]{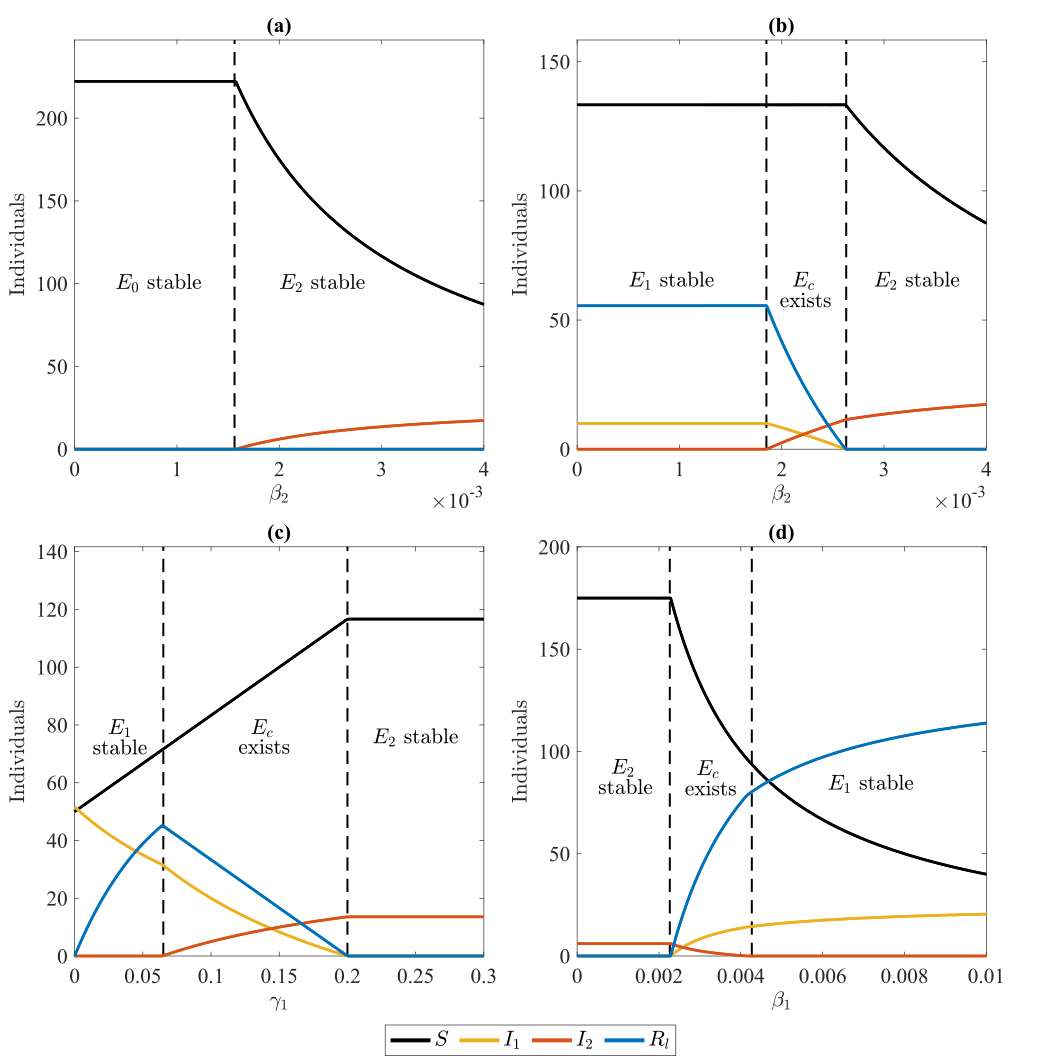}
\caption{{\small \textbf{(a)}: Bifurcation from disease free equilibrium ($E_0$) to dominance by strain 2 ($E_2$) using $\beta_2$ as a bifurcation parameter and where $\beta_1 = 0.0003$.  \textbf{(b)}: Bifurcation from strain 1 dominance ($E_1$) to coexistence ($E_c$) and finally dominance by strain 2 ($E_2$) using $\beta_2$ as a bifurcation parameter and where $\beta_1 = 0.003$. \textbf{(c)} Bifurcation from strain 1 dominance ($E_1$) to coexistence ($E_c$) and finally dominance by strain 2 ($E_2$) using $\gamma_1$ as a bifurcation parameter and where $\beta_1 = \beta_2 = 0.003$. \textbf{(d)} Bifurcation from strain 2 dominance ($E_2$) to coexistence ($E_c$) and then to dominance by strain 1 ($E_1$)  using $\beta_1$ as a bifurcation parameter where $\beta_2=.002$. All other parameter values are $\gamma_1=0.25$, $\gamma_2=0.2$, $d=.045$, $d_1=d_2=0.15$, $a=10$ and $\tau=0$.}  }
    \label{fig: bifurcation full model}
\end{figure}

\begin{figure}[!htbp]
\includegraphics[trim = 40mm 0mm 40mm 0mm, clip,width=\textwidth]{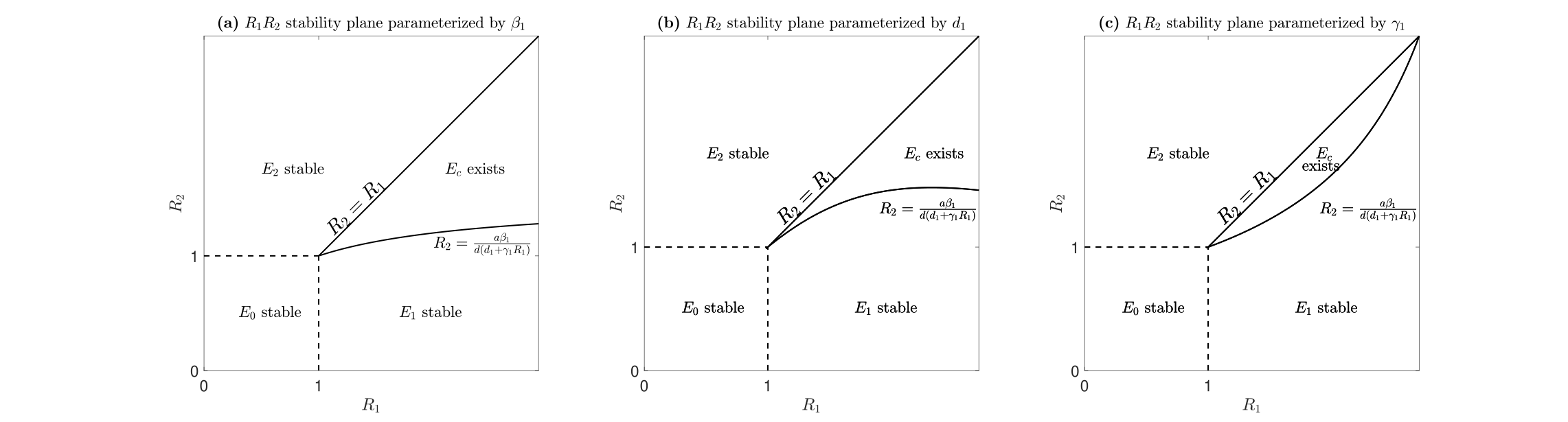}
    \caption{System \eqref{eqn: model1} stability and existence regions of equilibria in the $R_1R_2$ plane. (a) Parameterized by $\beta_1$. (b) Parameterized by $d_1$. (c) Parameterized by $\gamma_1$. For example, starting in the $E_0$ stability region and increasing $\beta_2$ ultimately produces a bifurcation as $E_0$ loses stability and $E_2$ gains stability. This is also illustrated with panel (a) of \autoref{fig: bifurcation full model}. Starting in the $E_1$ stability region and increasing $\beta_2$ produces a bifurcation as $E_1$ loses stability, $E_c$ appears and ultimately for higher $\beta_2$ values $E_2$ becomes stable. See panel (b) of \autoref{fig: bifurcation full model} for the bifurcation diagram. We note that the curve given by \eqref{eqn: R2 curve} can only be plotted as a function of $\beta_1, d_1$, $\gamma_1$, $a$ and $d$. We only show plots for the first three since the latter two have similar geometries.}
    \label{fig:R1R2 plane}
\end{figure}

\begin{figure}[htp]
    \includegraphics[width=\textwidth]{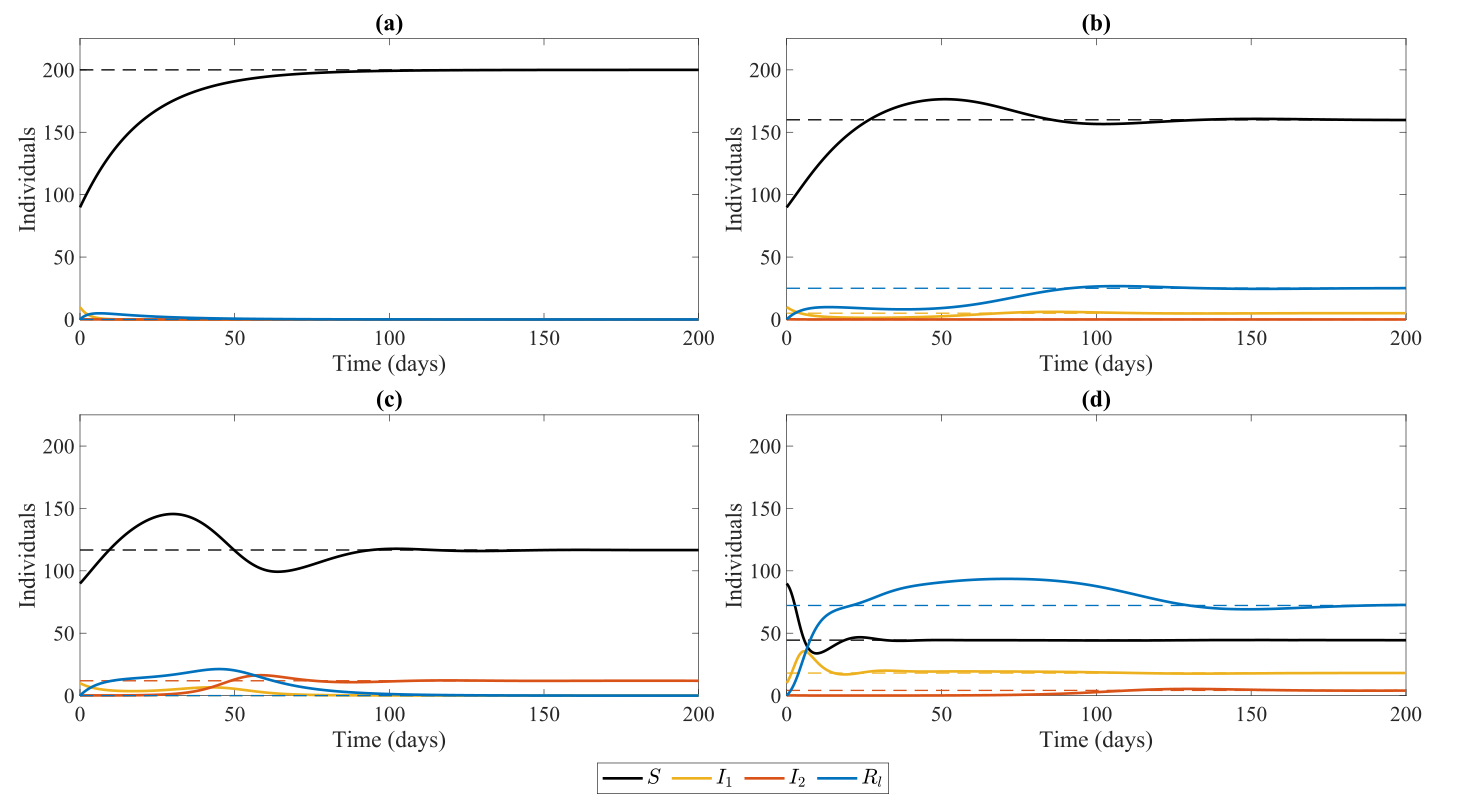}
    \caption{Model solutions illustrating the 4 equilibrium points of system \eqref{eqn: model1}. (a) disease free ($E_0$ stability region). Parameter values: $\beta_1=\beta_2=0.0003$. (b) strain 1 outcompetes strain 2 ($E_1$ stability region). $\beta_1=0.0025$ and $\beta_2=0.0003$. (c) strain 2 outcompetes strain 1 ($E_2$ stability region). $\beta_1=\beta_2=0.003$. (d) coexistence ($E_c$ existence region). $\beta_1=.009$ and $\beta_2=0.003$. Dashed lines represent equilibrium solutions (see equations \eqref{eqn: E_0},  \eqref{eqn: E_1}, \eqref{eqn: E_2} and \eqref{eqn: E_c}). The rest of the parameter values were fixed at $\gamma_1=0.25$, $\gamma_2=0.2$, $d=.05$, $d_1=0.15$, $d_2=0.15$, $a$=10 and $\tau=0$.}
    \label{fig:model equil}
\end{figure}

\autoref{fig:R1R2 plane} shows the general stability and existence regions of the equilibrium points of system \eqref{eqn: model1} in the $R_1R_2$ plane. Bifurcations occur when crossing from one region to another. To generate these diagrams we parameterize the curve \eqref{eqn: R2 curve} by either $\beta_1$, $d_1$ and $\gamma_1$. We illustrate the bifurcation from $E_0$ to $E_2$ in \autoref{fig: bifurcation full model} (panel (a)), the bifurcations from $E_1$ to $E_c$ to $E_2$ in (panel (b) and (c)), and the bifurcations from $E_2$ to $E_c$ to $E_1$ in (panel (d)).

We see that the competitive exclusion principle holds for either strain as long as the conditions of either Proposition \ref{LAS E1} or Proposition \ref{LAS E2} hold \citep{Gause1934,Bremermann1989}. However, conditions for one strain to competitively exclude the other are different between the two strains because of temporary cross-immunity. This also suggests that temporary cross-immunity is a mechanism for coexistence of two competing virus strains.

\section{The transient model}\label{sec:tran_model}
To study the stability of the coexistence steady state we make the assumption that the susceptible population is at equilibrium, $S(t)=S_{M}$ and remove the differential equation for $S$. Furthermore, if $I_1\neq0$, then $S_{M}=\frac{\gamma_1+d_1}{\beta_1}$. Therefore, $\frac{dI_1}{dt}=0$ and we may remove the equation for $\frac{dI_1}{dt}$, but assume $I_1(0)>0$. We have the following 2 dimensional system of differential equations:

\begin{equation}
\begin{aligned}
\frac{dI_2}{dt} & =  \beta_2 S_{M} I_2 +\beta_2 R_l I_2 -\gamma_2 I_2 -d_2I_2 &\\
\frac{dR_l}{dt} & = \gamma_1 I_1(0) - \beta_2 R_l I_2 - dR_l.
\end{aligned}
\label{eqn: trans_sys}
\end{equation} The reproduction numbers for the transient model are:

\begin{equation}
    \tilde{R}_1 = \frac{\beta_1}{d_1+\gamma_1} \text{    and     } \tilde{R}_2 = \frac{\beta_2}{d_2+\gamma_2} ,\text{respectively}.
\end{equation}
\subsection{Boundedness and positivity}
We prove basic positivity and boundedness of solutions for \eqref{eqn: trans_sys}. However, we note that if $\frac{\tilde{R}_1}{\tilde{R}_2}<1$, $I_2(t)$ becomes unbounded.
\begin{proposition}
Solutions to system \eqref{eqn: trans_sys} that start positive, remain positive for all time.
\label{th: positivity 2D}

\begin{proof}
Let $I_2(0)>0$ and $R_l(0)>0$. We proceed by way of contradiction. That is, supposed there exists $t_1>0$ where either $I_2(t_1)=0$ or $R_l(t_1)=0$ for the first time. Then for $t\in \left[0,t_1\right]$, we have that $I_2(t)>0$ and $R_l(t)>0$. We proceed by cases:

\textbf{Case 1:} $I_2(t_1)=0$

For $t\in\left[0,t_1\right]$, we have 
\begin{equation}
\begin{aligned}
I_2'(t)&\geq -\gamma_2I_2 - d_2I_2\\
     & = -\left(\gamma_2+d_2\right)I_2.
\end{aligned}
\end{equation}
This implies that 
$$I_2(t)\geq I_2(0)e^{-(\gamma_2+d_2)t}>0$$
Therefore, $I_2(t_1)>0$, a contradiction.

\textbf{Case 2:} $R_l(t_1)=0$

For $t\in\left[0,t_1\right]$, we have 
\begin{equation}
\begin{aligned}
R_l'(t)&\geq -\beta_2R_lI_2-dR_l\\
     & \geq -\beta_2R_l\alpha-dR_l.
\end{aligned}
\end{equation}
where \[\alpha:=\max_{t\in[0,t_1]}\{I_2(t)\}.\] Then \begin{equation}
\begin{aligned}
R_l'(t)&\geq -\beta_2R_lI_2-dR_l\\
     & \geq -\left(\beta_2\alpha+d\right)R_l.
\end{aligned}
\end{equation}
This implies that 
$$R_l(t)\geq R_l(0)e^{-(\beta_2\alpha+d)t}>0$$
Therefore, $R_l(t_1)>0$, a contradiction.
\end{proof}
\end{proposition}
\begin{proposition}
\label{prop: trans bounded}
If $\frac{\tilde{R}_1}{\tilde{R}_2}>1$, then solutions of system \eqref{eqn: trans_sys} are bounded from above.
\label{th: boundedness 2D}\begin{proof}
Assume that $\frac{\tilde{R}_1}{\tilde{R}_2}>1$, then $\gamma_2+d_2-\beta_2S_M>0$. Let $z=I_2+R_l$, $z_0=I_2(0) + R_l(0)$ and $\alpha = \min\{\gamma_2+d_2-\beta_2S_M,d\}$. Then \begin{equation*}
\begin{aligned}
    z'(t)&= \gamma_1I_1(0) -\left(\gamma_2+d_2-\beta_2S_M\right)I_2 -dR_l\\
     &\leq  \gamma_1I_1(0) - \alpha z. 
\end{aligned}
\end{equation*}
This implies $$z(t)\leq \frac{\gamma_1 I_1(0)}{\alpha } + \left(z_0 -\frac{\gamma_1I_1(0)}{\alpha}\right)e^{-\alpha t}$$and $$\limsup_{t\rightarrow\infty}z(t)\leq
\frac{\gamma_1 I_1(0)}{\alpha}.$$Therefore, $$z(t)\leq
\max\left\{z_0,\frac{\gamma_1 I_1(0)}{\alpha}\right\}=:B.$$Thus we have $I_2+R_l\leq B$. Since $I_2>0$ and $R_l>0$, we have that both $I_1$ and $R_l$ are bounded above.
\end{proof}
\end{proposition}Lastly, we see that when $\frac{\tilde{R}_1}{\tilde{R}_2}<1$, then solutions for $I_2$ are unbounded.
\begin{proposition}
\label{thm: I2 unbounded}
If $\frac{\tilde{R}_1}{\tilde{R}_2}<1$ then $I_2$ from system \eqref{eqn: trans_sys} is unbounded.
\begin{proof}
We have $\frac{\tilde{R}_1}{\tilde{R}_2}<1$, then
\begin{equation*}
\begin{aligned}
    I_2'(t)&= \left(\beta_2S_M - \gamma_2 -d_2\right)I_2 +\beta_2R_lI_2 \\
     &\geq \left(\beta_2S_M - \gamma_2 -d_2\right)I_2\\
     &=\frac{\beta_2}{\tilde{R}_1}\left(1-\frac{\tilde{R}_1}{\tilde{R}_2}\right)I_2>0.
\end{aligned}
\end{equation*}
This implies that $I_2(t)$ is unbounded for all $t>0$.
\end{proof}
\end{proposition}An interesting implication of Proposition \ref{thm: I2 unbounded} is that if $I_1$ is not as infectious relative to strain 2, then it cannot control the spread of strain 2 and ultimately strain 2 becomes unbounded in the transient model.
\subsection{Equilibria of the transient system}
For our analysis we would like to have bounded solutions. For this to hold, by Proposition \ref{th: boundedness 2D} we must have that $\frac{\tilde{R}_1}{\tilde{R}_2}>1$. Therefore, for the remainder of this section we assume that $\frac{\tilde{R}_1}{\tilde{R}_2}>1$. 

Assuming that $I_1(0)> 0$, we find two equilibria: the coexistence equilibria, $U_c$, and another equilibrium where the first strain exists, $U_1$. They take the following form:
 
\begin{equation}
    U_c = \left( \frac{\beta_1\gamma_1I_1(0)-d\left(d_1+\gamma_1\right)\left(\frac{\tilde{R}_{1}}{\tilde{R}_{2}}-1\right)}{
    \beta_2\left(d_1+\gamma_1\right)\left(\frac{\tilde{R}_{1}}{\tilde{R}_{2}}-1\right)},\frac{d_1+\gamma_1}{\beta_1}\left(\frac{\tilde{R}_1}{\tilde{R}_2}-1\right)\right)
\end{equation}  
\vspace{.5cm}
\begin{equation}
    U_1 = \left( 0,\frac{\gamma_1 I_1(0)}{d}\right).
\end{equation}  
We note that $U_c$ is dependent on $I_1(0)$ and is biologically relevant exactly when \begin{equation}
    \beta_1\gamma_1I_1(0)-d\left(d_1+\gamma_1\right)\left(\frac{\tilde{R}_1}{\tilde{R}_2}-1\right)>0. 
    \label{eqn: U_c existence}
\end{equation}  We note that $U_c$ can exist even when both reproduction numbers are less than 1. We have the following theorem on the stability of $U_c$.
\begin{proposition}
If $U_c$ exists, then it is asymptotically stable.
\label{th: Uc LAS}
\begin{proof}
$U_c$ is a positive steady state if and only if $$\beta_1\gamma_1I_1(0)-d\left(d_1+\gamma_1\right)\left(\frac{\tilde{R}_1}{\tilde{R}_2}-1\right)>0.$$
The Jacobian matrix at $U_c$ is

\begin{equation*}
\begin{pmatrix}
0 & \frac{\beta_1\gamma_1I_1(0)-d\left(d_1+\gamma_1\right)\left(\frac{\tilde{R}_1}{\tilde{R}_2}-1\right)}{
\left(d_1+\gamma_1\right)\left(\frac{\tilde{R}_1}{\tilde{R}_2}-1\right)}\\
-\frac{\beta_2}{\beta_1}\left(d_1+\gamma_1\right)\left(\frac{\tilde{R}_1}{\tilde{R}_2}-1\right) & \frac{-\beta_1\gamma_1I_1(0)}{\left(d_1+\gamma_1\right)\left(\frac{\tilde{R}_1}{\tilde{R}_2}-1\right)}
\end{pmatrix}.
\end{equation*}
We find that the trace is
\begin{equation*}
\frac{-\beta_1\gamma_1I_1(0)}{\left(d_1+\gamma_1\right)\left(\frac{\tilde{R}_1}{\tilde{R}_2}-1\right)} <0
\end{equation*}
and determinant is
\begin{equation*}
\frac{\beta_2}{\beta_1}\left(\beta_1\gamma_1I_1(0) - d\left(d_1+\gamma_1\right)\left(\frac{\tilde{R}_1}{\tilde{R}_2}-1\right)\right)>0.
\end{equation*}
Therefore, both eigenvalues have negative real part and $U_c$ is locally asymptotically stable whenever it exists.
\end{proof}
\end{proposition}We have the following theorem on the stability of $U_1$.
\begin{proposition}
\label{th: U1 LAS}
$U_1$ always exists. Furthermore,
\begin{enumerate}
    \item  $U_1$ is locally asymptotically stable when $\beta_1\gamma_1I_1(0) -d\left(d_1+\gamma_1\right)\left(\frac{\tilde{R}_1}{\tilde{R}_2}-1\right)<0$. In addition, $U_c$ does not exist.
    \item $U_1$ is unstable when $\beta_1\gamma_1I_1(0) -d\left(d_1+\gamma_1\right)\left(\frac{\tilde{R}_1}{\tilde{R}_2}-1\right)>0$.
\end{enumerate}

\begin{proof}
The Jacobian matrix at $U_1$ is 

\begin{equation*}
\begin{pmatrix}
\frac{\beta_2\left(d_1+\gamma_1\right)}{\beta_1} - \left(d_2+\gamma_2\right) + \frac{\beta_2 \gamma_1 I_1(0)}{d} & 0\\
-\frac{\beta_2 \gamma_1 I_1(0)}{d} & -d
\end{pmatrix}.
\end{equation*}
We find that the eigenvalues are

\begin{equation}
    \lambda_1 = \frac{\beta_2}{\beta_1 d}\left(\beta_1\gamma_1I_1(0) -d\left(d_1+\gamma_1\right)\left(\frac{\tilde{R}_1}{\tilde{R}_2}-1\right)\right)
\end{equation}
\begin{equation*}
\lambda_2 = -d
\end{equation*}
We see that $\lambda_2<0$ exactly when $\beta_1\gamma_1I_1(0) -d\left(d_1+\gamma_1\right)\left(\frac{\tilde{R}_1}{\tilde{R}_2}-1\right)<0$ and unstable when $$\beta_1\gamma_1I_1(0) -d\left(d_1+\gamma_1\right)\left(\frac{\tilde{R}_1}{\tilde{R}_2}-1\right)>0.$$
\end{proof}
\end{proposition}


\begin{theorem}
\label{thm: global stability Uc}
If $\beta_1\gamma_1I_1(0) -d\left(d_1+\gamma_1\right)\left(\frac{\tilde{R}_1}{\tilde{R}_2}-1\right)>0$, then all solutions tend to $U_c$.
\begin{proof}
To simplify our calculation we let $x=I_2$ and $y=R_l$. Furthermore, let $a=\gamma_1 I_1(0)$,  $b=\gamma_2+d_2-\beta_2S_M$ and $\beta=\beta_2$. Then system \ref{eqn: trans_sys} becomes
\begin{equation}
\begin{aligned}
\frac{dx}{dt} & =  \beta xy -bx &\\
\frac{dy}{dt} & = a - \beta xy -dy.
\end{aligned}
\label{eqn: trans_sys x_y}
\end{equation}
With equilibrium solution $(x^*,y^*)=U_c$. The system is the exact same as system \eqref{eqn: limiting system} and hence a Lyapunov function is
\begin{equation*}
    V(x,y)=x-x^*\ln(x)+y-y^*\ln(y).
\end{equation*}
\end{proof}

\end{theorem}

A phase portrait of the solution trajectory to the coexistence steady state is shown in \autoref{fig:phase_portrait}. Furthermore, it can be shown that $U_1$ is globally asymptotically stable under certain conditions.

\begin{theorem}\label{thm:U1gas}
If $\beta \left(\frac{a}{d}\right)<b$, then all solutions tend to $U_1$.
\begin{proof}
We prove this result by contradiction. Recall that if $\beta\left(\frac{a}{d}\right)<b$, then the transient system \eqref{eqn: trans_sys x_y} does not attain a positive steady state.  Assume that $\beta\left(\frac{a}{d}\right)<b$ and $\lim_{t\to\infty} (x,y) \neq (0,a/d)$. Furthermore, observe that $\limsup_{t\to\infty}y(t) \leq a/d$. Thus for $\varepsilon>0$, there exists a $t_*>0$ such that $y(t)<a/d+\varepsilon$ for $t>t_*$. With this claim, we see that, for $t>t_*$,
    $$\frac{dx}{dt} = \beta yx - bx < \beta\left(\frac{a}{d}+\varepsilon\right)x - bx<0,$$
implying that $\lim_{t\to\infty}x(t) = x_* \geq0$. If $x_*>0$, then an application of Barbalat's lemma \citep{barbalat} yields
    $$\lim_{t\to\infty}\frac{dx}{dt} = 0 = \beta x_* \left( \lim_{t\to\infty}y(t) \right) - bx_*,$$
which shows that $\lim_{t\to\infty}y(t) = \frac{b}{\beta} >0$. Hence we obtain the positive steady state $E_* = (x_*,\frac{b}{\beta})$, which contradicts the fact that the model \eqref{eqn: trans_sys x_y} has no positive steady state. In other words, the claim yields $\lim_{t\to\infty}x(t) = 0$.

Furthermore, since $y$ is bounded, the above result implies that for any $\varepsilon_1<a$, there exists a  $t_1>t_*$ such that $\beta xy<\varepsilon_1$ for $t>t_1$. Therefore
    $$\frac{dy}{dt}\geq a - \varepsilon_1 - dy$$
for $t>t_1$, yielding
    $$\liminf_{t\to\infty}y(t) \geq \frac{a-\varepsilon_1}{d}.$$
Letting $\varepsilon_1\to0$, we see that $\liminf_{t\to\infty}y(t) \geq a/d$. As well, our claim indicates that $\limsup_{t\to\infty}y(t) \leq a/d$. Hence $\lim_{t\to\infty} y(t) = a/d$.

In the following, we prove our claim. The proof is divided into three cases:
    \begin{enumerate}
        \item $y(0)\leq a/d$;
        \item $y(0) > a/d$ and there exists a $t_2>0$ such that $y(t) > a/d$ for $t\in[0,t_2)$ and $y(t_1) = a/d$;
        \item $y(t)>a/d$ for all $t>0$.
    \end{enumerate}
We consider case 1. We have
    $$\frac{dy}{dt} < a-dy = d\left(\frac{a}{d} - y \right) \implies y(t) < \frac{a}{d} + \left(y(0) - \frac{a}{d}\right)e^{-dt}.$$
Hence $y(t)<a/d$ for $t>0$, and our claim is true.

Consider the second case. From case 1, we see that $y(t) < a/d$ for $t>t_1$ and again our claim is true.

Finally, consider case 3. Here, $dy/dt<0$ and there is a $y_c\geq a/d$ such that 
    \begin{equation}\label{eq:star1}\tag{$*$}
        \lim_{t\to\infty} y(t) = y_c \geq \frac{a}{d}.
    \end{equation}
By Barbalat's lemma \citep{barbalat}, we have
    $$0 = \lim_{t\to\infty} (a - \beta xy - dy) = a - \beta y_c \left(\lim_{t\to\infty} x(t) \right) - dy_c$$
which implies
    $$y_c = \frac{a}{d + \beta\lim_{t\to\infty}x(t)} \leq \frac{a}{d}.$$
This together with \eqref{eq:star1} imply
    $$\lim_{t\to\infty} y(t) = \frac{a}{d} \quad \text{and} \quad \lim_{t\to\infty}x(t) = 0,$$
contradicting the assumption that $\lim_{t\to\infty} (x,y) \neq (0,a/d)$. This concludes the proof.
\end{proof}
\end{theorem}

\begin{figure}[!htb]
    \centering
\includegraphics[width=4in]{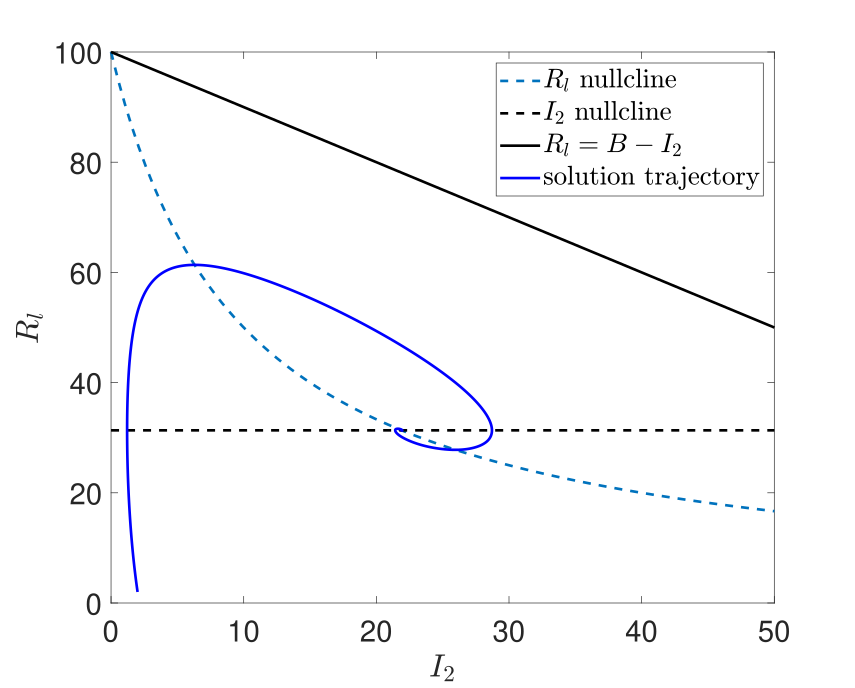}
    \caption{Phase plane with solution trajectory and coexistence steady state $U_c$ of system \eqref{eqn: trans_sys}. Parameter values: $\beta_1=0.03$, $\beta_2=0.01$ $\gamma_1=0.1$, $\gamma_2=0.28$, $d=d_1=d_2=0.1$ and $a=0.7$. Here $B$ is defined as in Proposition \ref{prop: trans bounded}.}
    \label{fig:phase_portrait}
\end{figure}

From Propositions \ref{th: Uc LAS} and \ref{thm: global stability Uc} we see that both virus strains can coexist as long as the original strain has a higher reproduction number than strain 2 and $$\beta_1\gamma_1I_1(0)-d\left(d_1+\gamma_1\right)\left(\frac{\tilde{R}_1}{\tilde{R}_2}-1\right)>0.$$ However, we may solve for $\tilde{R}_2$ in terms of $\tilde{R}_1$ to generate a bifurcation curve between coexistence and competitive exclusion,
\begin{equation}
    \tilde{R}_2=\frac{\tilde{R}_1}{\frac{\gamma_1I_1(0)}{d}\tilde{R}_1+1}.
    \label{eqn: bif curve trans}
\end{equation}
\autoref{fig:R02R02 plane} shows the $\tilde{R}_2\tilde{R}_1-$bifurcation plane where equation \eqref{eqn: bif curve trans} is parameterized by $\beta_1$ $\gamma_1$ or $d_1$. For the two strains to coexist together, strain 1 needs to have a higher basic reproduction number than strain 2. However, it can't be too high relative to strain 2 or it will force strain 2 to extinction. The unbounded region corresponds to Proposition \ref{thm: I2 unbounded}. In general, the model suggests that viruses which mutate into strains that are slightly less infectious are more likely to coexist together. On the other hand, viruses that mutate into strains that are sufficiently less infectious relative to the original strain, will out-compete the mutated strain.
\begin{figure}[!htbp]
    \includegraphics[trim = 60mm 40mm 60mm 40mm, clip,width=\textwidth]{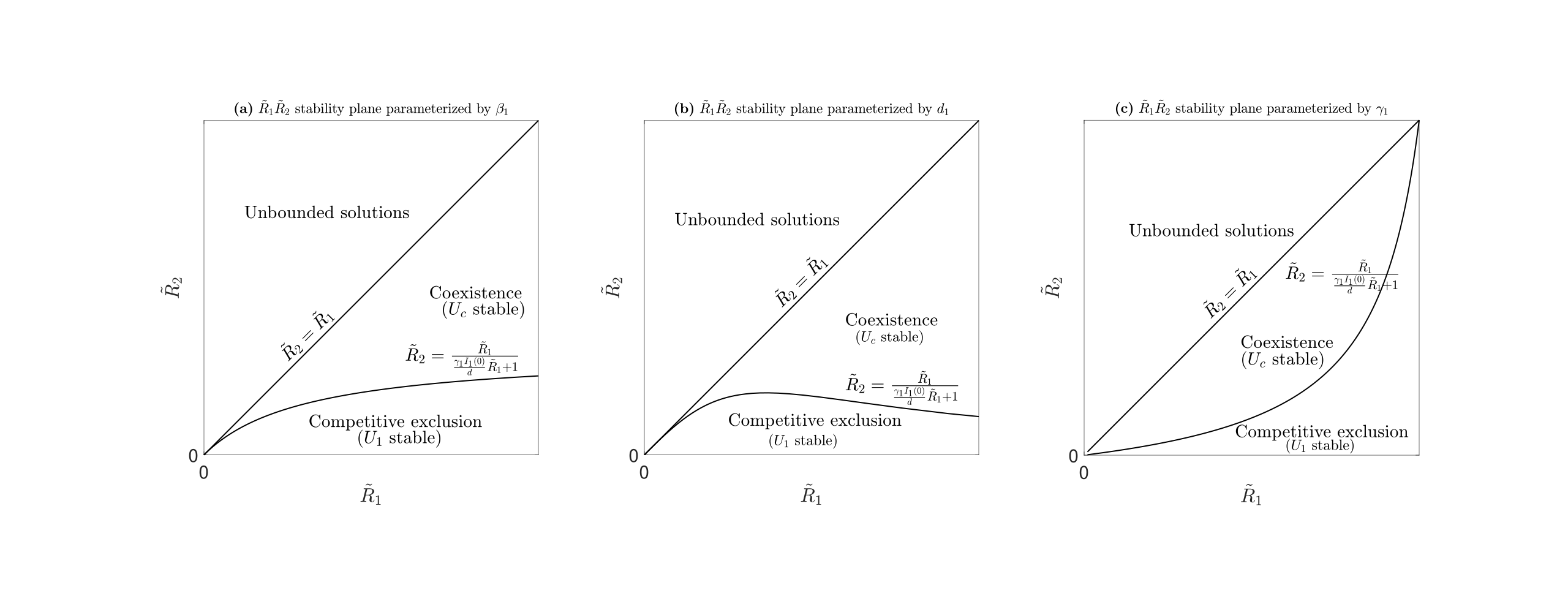}
    \caption{The $\tilde{R}_2\tilde{R}_1$-plane for the transient model (system \ref{eqn: trans_sys}). The model exhibits 3 different dynamics: 1) unbounded solutions, 2) coexistence of the two virus strains and 3) competitive exclusion of the 2nd strain. \textbf{(a)} Stability plane parameterized by $\beta_1$; \textbf{(b)} Stability plane parameterized by $d_1$; \textbf{(c)} Stability plane parameterized by $\gamma_1$; If starting in the competitive exclusive region where $I_1$ is the long-term winner and then increasing $\beta_2$ (thus increasing $\tilde{R}_2$) we see that a bifurcation occurs when $\tilde{R}_2=\frac{\tilde{R}_1}{\frac{\gamma_1I_1(0)}{d}\tilde{R}_1}+1$ and the second strain can coexist with the first strain. Increasing $\beta_2$ even more will ultimately lead to another bifurcation where $I_2$ becomes unbounded (by Proposition \ref{thm: I2 unbounded}). }
    \label{fig:R02R02 plane}
\end{figure}

\begin{figure}[!htbp]
\includegraphics[trim = 20mm 0mm 20mm 0mm, clip,width=\textwidth]{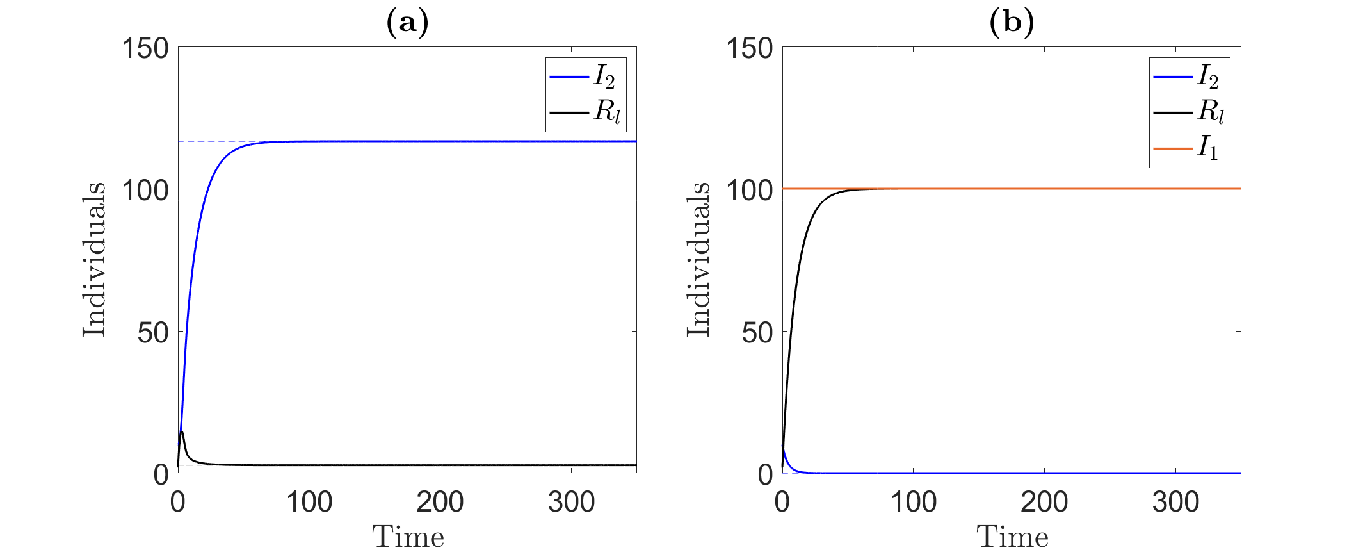}
    \caption{Long-term dynamics of system \eqref{eqn: trans_sys}. \textbf{(a)} coexistence steady state where $I_1$ and $I_2$ coexist. \textbf{(b)} competitive exclusion of $I_2$ by $I_1$. We do not plot $S$ or $I_1$ since they are held constant at $S_M$ and $I_1(0)$, respectively.}
    \label{fig:trans asymp}
\end{figure}

\begin{table}[!htbp]
\begin{center}
\caption{Established results and open questions}\label{tab: results}%
\begin{tabular}{@{}lll@{}}

\toprule
& Conditions  & Results or question \\
\midrule
1. &$R_i>1$    &  Existence of strain-specific equilibrium $E_i$, $i=1,2$  \\

2. &$\frac{a\beta_1}{d\left(\gamma_1R_2+d_1\right)}R_2<R_1<\frac{a\beta_1}{dd_1}$    &  Existence of coexistence equilibrium  \\
3. &$\max\left\{R_1,R_2\right\}<1$    &  Disease free equilibrium, $E_0$ is globally stable  \\
4. &See theorem \ref{GS E1} \ &  $E_1$ is globally stable \\
5. &$R_1<1$ and $R_2>1$ & $E_2$ is globally stable  \\
6. & See \cite{fudolig_2020} & Local stability of $E_c$\\
7. & Open & Global stability of $E_c$\\
8. & Open & Global stability of $E_1$ with $R_2>1$.\\
9. & Open & Global stability of $E_2$ with $R_1>1$.\\
10. & Open & Influence of $\tau$ on the stability of $E_c$.\\
\midrule
1.& $\frac{\gamma_1I_1(0)\beta_2}{d}<\gamma_2+d_2-\beta_2S_M$ & $U_1$ is globally stable\\
2.& Inequality \eqref{eqn: U_c existence} & $U_c$ is globally stable\\
\botrule
\end{tabular}
\end{center}
\end{table}

\section{Numerical results}\label{sec:num}
\subsection{Data fitting}
 The system \eqref{eqn: model1} is validated by fitting to wastewater data from October 1, 2020 to May 13, 2021 obtained from the Deer Island Treatment Plant in Massachusetts \citep{xiao_2022}. This plant serves approximately 2.3 million people in the greater Boston area \citep{xiao_2022}. More information on the collection and processing of wastewater samples can be found in \cite{xiao_2022}. Fitting to wastewater data, as opposed to incidence or mortality data, allows us to avoid underreporting issues related to clinical reporting.

The B.1.1.7 (Alpha) variant was detected in Massachusetts in January 2021 \citep{alpha_MA}, while the B.1.617 (Delta) variant was found in the state in April 2021 \citep{delta_MA}. It should be noted that Massachusetts (population size 7 million) began vaccinating healthcare workers on December 15, 2020 during Phase 1 of the state's vaccination plan \citep{vax_phase}. For simplification purposes, we assume that individuals in the susceptible ($S$) and recovered ($R_l$) compartments are vaccinated at a rate $v$ and that the vaccine offers immediate protection from both strains.  Individuals who have recovered from the emerging strain are not tracked or vaccinated for several reasons. In the presented model, these individuals are removed from the population and thus do not impact infection dynamics. It has been shown that two vaccine doses provided significant protection against the Alpha and Delta variants with respect to infection and hospitalization \citep{gram_2022}. Although protection against infection has been found to wane over time, \cite{gram_2022} found that, after 120 days, vaccine efficacy against Delta decreased from 92.2\% to 64.8\% in those aged 12 to 59 years. Vaccine efficacy in individuals over 60 years of age saw decreases in efficacy from 90.7\% to 73.2\% and 82.3\% to 50\% for Alpha and Delta, respectively \citep{gram_2022}. Due to the limited time-scale of vaccination in the model and the scope of this study, we assume protection does not wane.

Based on data on fully-vaccinated individuals (defined as those who received all doses of the vaccine protocol) from the U.S. Centers for Disease Control and Prevention (U.S. CDC), compiled by Our World in Data \citep{vax_data,vax_cdc}, we fix the per capita vaccination rate at $v = 0.0038$ per day with vaccination beginning on January 5, 2021 due to the three week time period between first and second doses \citep{vax_phase}. The calculation of $v$ is shown in \autoref{fig:vax_fit}.

\begin{figure}[!htbp]
    \centering
    \includegraphics[width=\textwidth]{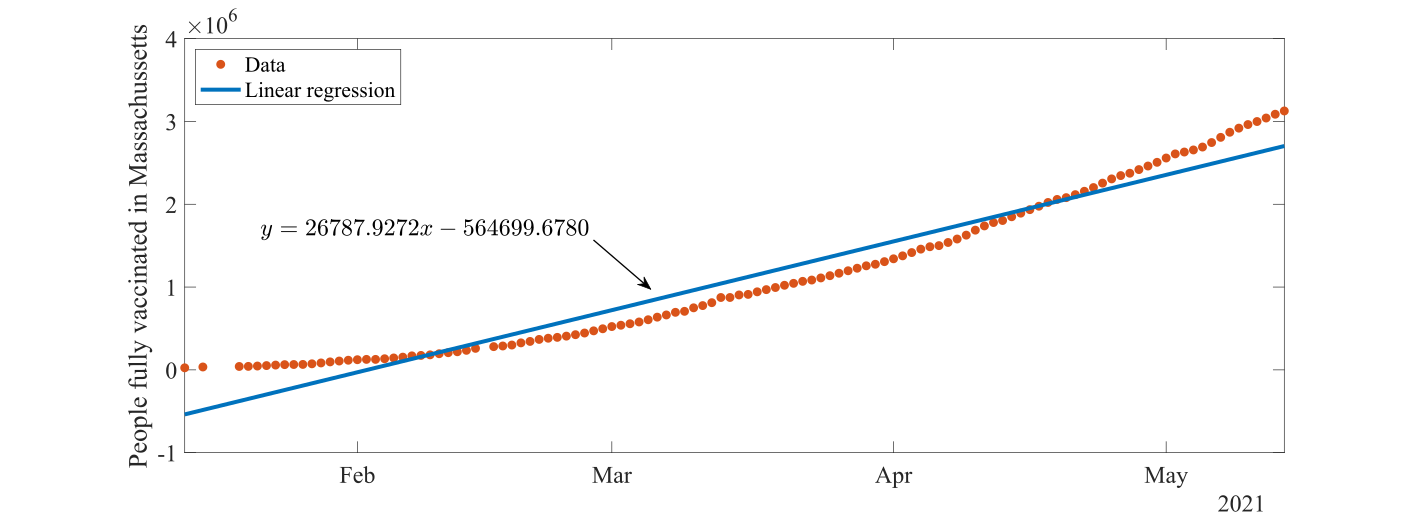}
    \caption{Line of best-fit (blue) compared to Massachusetts vaccination data. The per capita vaccination rate $v = 0.0038$ is given by the slope of the best-fit line divided by the total state population of 7 million.}
    \label{fig:vax_fit}
\end{figure}

In order to fit system \eqref{eqn: model1} to the wastewater data, we add a compartment $C_V(t)$ denoting the cumulative viral RNA copies in the wastewater following the formulations in \cite{saththasivam_2021}. Hence, the dynamics of the cumulative virus released into the wastewater is governed by
    $$\frac{dC_V}{dt} = \alpha \delta (1-\eta) \left(I_1 + I_2\right),$$
where $\alpha$ denotes the fecal load per individual in grams per day, $\delta$ denotes the viral shedding rate per gram of stool, and $(1-\eta)$ denotes the proportion of RNA that arrives to the wastewater treatment plant. Because the wastewater data is daily and $C_V(t)$ is cumulative, the objective function be minimized is given by
    $$\text{SSE} = \sum_{t_n} \left(\log_{10}(\tilde{C}_V(t_n)) - \log_{10}{\text{data}(t_n)} \right)^2$$
where $\tilde{C}_V(t_n) = C_V(t_n) - C_V(t_{n-1})$ (i.e., new viral RNA entering the sewershed on day $t_n$). Parameter estimation is carried out using \texttt{fmincon} and 1000 \texttt{MultiStart} runs in Matlab. For comparison purposes, both the ODE and DDE versions of the model were fit to the data. Initial values for $I_1$ and $I_2$ are estimated by using the initial viral RNA data and the estimated values of $\alpha$, $\delta$, and $\eta$; that is, the constraint
    $$I_1(0) + I_2(0) = \frac{\text{initial viral data}}{\alpha\times\delta\times(1-\eta)},$$
and assuming that $I_1(0)\geq I_2(0)$. For the model with time delay, the same constraints are used for the initial histories. Values for estimated and fixed parameters are listed in \autoref{tab:param}.

\autoref{fig:ODE_fit} depicts model simulations without time delay using the best-fit parameters when compared to daily wastewater data (\autoref{fig:ODE_fit}a) and seven-day average case data (\autoref{fig:ODE_fit}b). The ODE version of the model predicts peak new infections on December 29, 2020, preceding the daily reported case data by 11 days. Due to the unreliability in the case data, however, this 1.5 week difference may be reasonable. Furthermore, the model projects approximately six times more new cases than the reported case data at their respective peaks.

Best-fit simulations with time delay are shown in \autoref{fig:DDE_fit}. Here, the model predicts daily incidence peaking on January 4, 2021, approximately five times higher than the reported cases on January 9, 2021, a difference of 5 days. Unlike the ODE version, the inclusion of time delay allows the model to capture the decline of the Alpha wave, but both the ODE and DDE versions of the model are unable to capture the Delta wave.

\begin{figure}[!htbp]
    \centering
    \includegraphics[width=\textwidth]{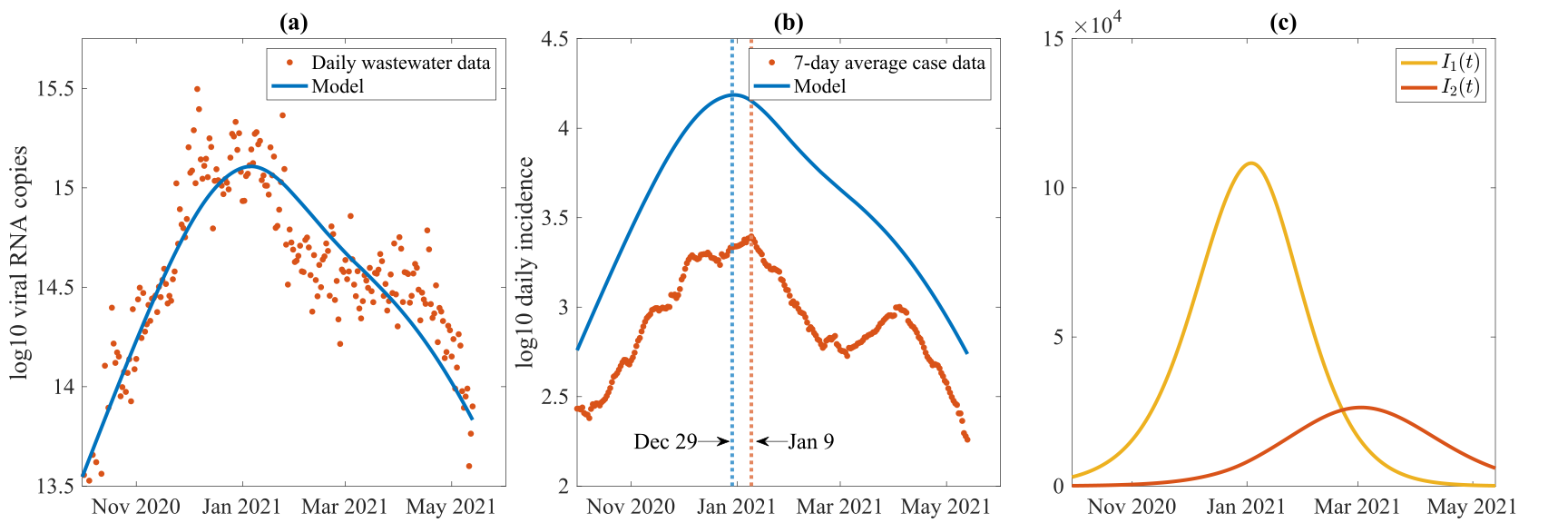}
    \caption{\textbf{(a)} Best-fit model without time delay compared to wastewater data (SSE = 8.6230). \textbf{(b)} Model output of total daily new cases of $I_1$ and $I_2$ compared to seven-day average of new reported cases. Dotted lines indicate date of maximum reported cases for the data (orange) and the model (blue). \textbf{(c)} Model output of strain 1 (solid line) and strain 2 (dashed) line over time}
    \label{fig:ODE_fit}
\end{figure}

\begin{figure}[!htbp]
    \centering
    \includegraphics[width=\textwidth]{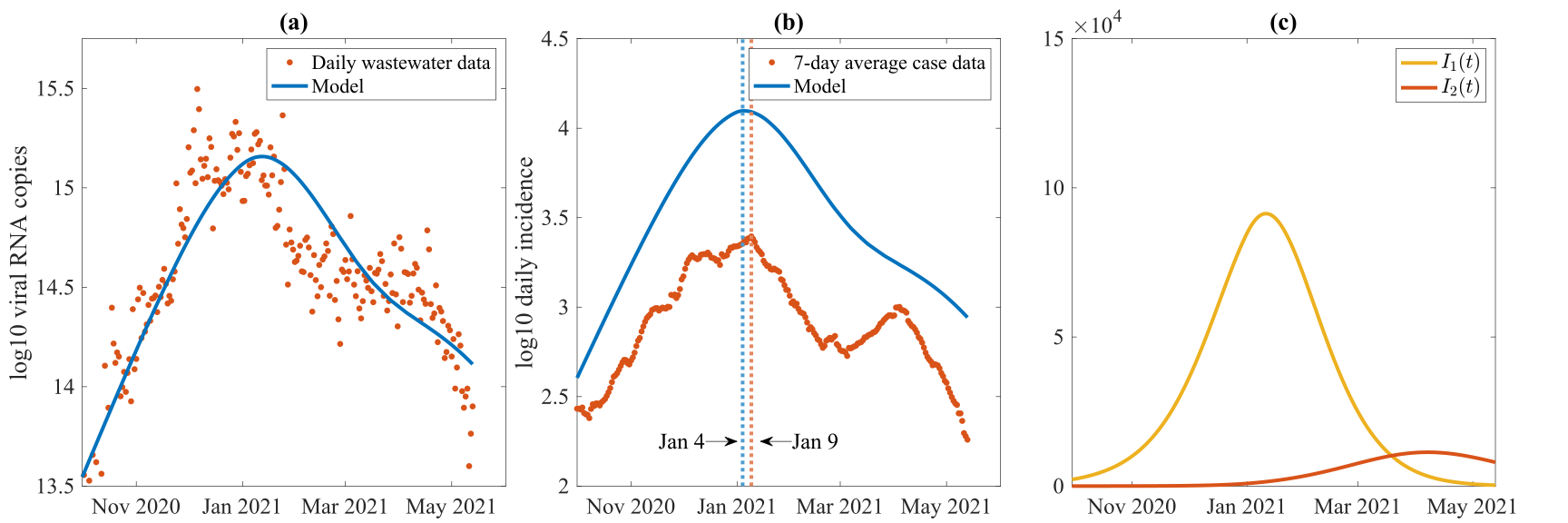}
    \caption{\textbf{(a)} Best-fit model with estimated time delay compared to wastewater data (SSE = 11.0819). \textbf{(b)}  Model output of total daily new cases of $I_1$ and $I_2$ compared to seven-day average of new reported cases. Dotted lines indicate date of maximum reported cases for the data (orange) and the model (blue). \textbf{(c)} Model output of strain 1 (solid line) and strain 2 (dashed) line over time}
    \label{fig:DDE_fit}
\end{figure}

\begin{sidewaystable}[!htbp]
\begin{center}
\begin{minipage}{550pt}
\caption{Parameters of the ODE and DDE versions of the model \eqref{eqn: model1}. All units are day$^{-1}$ unless otherwise noted}\label{tab:param}%
\begin{tabular}{@{}lllll@{}}
\toprule
Parameter & Description  & ODE & DDE & Reference\\
\midrule
$a$    & Birth rate (persons per day) & 62.1  & 62.1 &  M.A. Department of Public Health (2022)\footnote{\cite{MA_births}}\\
$d$    & Natural death rate   &  0.000035  & 0.000035 & \cite{MA_deaths}\\
$v$ & Vaccination rate & 0.0038 & 0.0038 & \cite{vax_data}, U.S. CDC (2022)\footnote[2]{\label{b}\cite{vax_cdc}}\\
$\alpha$ & Fecal load (grams per day per person) & 149 & 149 & \cite{saththasivam_2021}\\
$\delta$ & Viral shedding rate (copies/gram) & 1.3561$\times10^8$ & 2.0733$\times10^8$ & Fitted\\
$\eta$ & Losses in the sewer (unitless) & 0.4755 & 0.5005 & Fitted \\
\midrule
$\beta_1$   & Strain 1 contact rate (per person per day)& 7.602$\times10^{-8}$ & 7.7330$\times10^{-8}$ & Fitted \\
$d_1$ &  Strain 1 disease-induced mortality rate & 0.001 & 0.0044 & Fitted\\
$\gamma_1$ & Strain 1 recovery rate & 1/8 & 1/8 & \cite{killingley_2022}, U.S. CDC (2022)\footnotemark[2] \\
\midrule
$\beta_2$ & Strain 2 contact rate (per person per day)& 7.6516$\times10^{-8}$ & 8.2870$\times10^{-8}$ & Fitted\\
$d_2$ & Strain 2 disease-induced mortality rate & 0.0009 & 0.00001 & Fitted\\
$\gamma_2$ & Strain 2 recovery rate & 1/8 & 1/8 & \cite{killingley_2022}, U.S. CDC (2022)\footnotemark[2]\\
$\tau$ & Temporary cross-immunity (days) & & 2.0358 & Fitted\\
\midrule
$I_1(0)$ & Initial individuals infected with strain 1 & 3028.3530 & 2208.2008 & Fitted\\
$I_2(0)$ & Initial individuals infected with strain 2 & 110.1813 & 2.6861 & Fitted\\
SSE & & 8.6230 & 11.0819\\
\botrule
\end{tabular}
\end{minipage}
\end{center}
\end{sidewaystable}

\subsection{Sensitivity analysis}
In this section, we carry out a local sensitivity analysis to explore which parameters are the most important to model dynamics. We use a normalized sensitivity analysis so that the sensitivity coefficients are not affected by parameter magnitude. Here, the normalized sensitivity coefficients are given by \citep{saltelli_2000}:
    $$s_p = \frac{\partial Y}{\partial p}\times\frac{p}{Y} \approx \frac{[Y(p+\Delta p) - Y(p)]/Y(p)}{\Delta p/p}, $$
where $p$ and $Y$ denote the parameter and response of interest, respectively, and $\Delta p$ is the perturbation size. Each parameter is varied by 1\% individually from the values listed in \autoref{tab:param} while all other parameters are fixed. Here, the response variable $Y$ is cumulative cases evaluated at steady state. We ignore the parameters related to wastewater ($\alpha$, $\delta$, and $\eta$) since they do not impact disease dynamics in the analysis. Results are shown in \autoref{fig:sens}. The height of the bars indicates how sensitive the response variable is to the parameter; the direction of the bars (or sign of the sensitivity coefficient) indicates the direction of correlation.

The ODE and DDE versions of the model display significant sensitivity to the strain-specific contact rates ($\beta_1,\beta_2$) and the strain-specific recovery rates ($\gamma_1,\gamma_2$); the DDE version of the model has increased sensitivity to the initial number of those infected with strain 1 compared to the model without time delay. Furthermore, model dynamics, independent of time delay, are only slightly (if at all) impacted by changes in the strain-specific mortality rates ($d_1,d_2$).

\begin{figure}[!htbp]
    \centering
    \includegraphics[width=\textwidth]{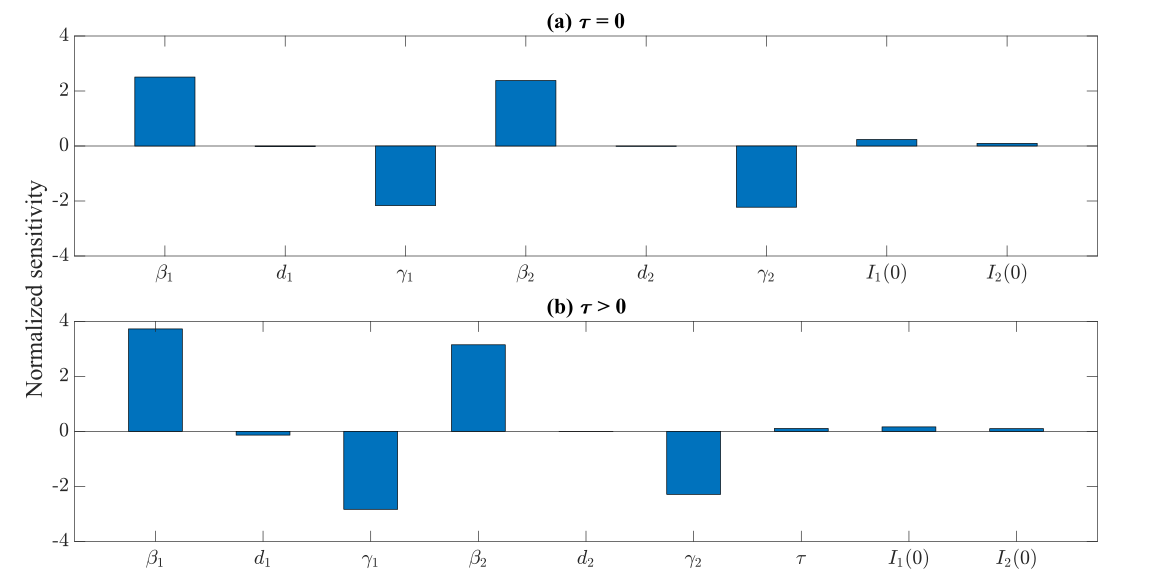}
    \caption{Local normalized sensitivity analysis with respect to cumulative steady state cases for the model \eqref{eqn: model1} \textbf{(a)} without time delay, and \textbf{(b)} with time delay. Parameters are varied by 1\% one at a time. Baseline values are listed in \autoref{tab:param}}
    \label{fig:sens}
\end{figure}

\section{Discussion}

In this paper, we have constructed a mathematical model describing two-strain virus dynamics with temporary cross-immunity. Although this general framework is applicable to many diseases, we put our model into the context of the COVID-19 pandemic and connected infectious individuals with wastewater data. The model produces rich long-term dynamics that include: (1) a state where the two strains are not infectious enough and are cleared from the population, (2) two competitive exclusion states where one of the strains is more infectious than the other and ultimately forces the other to extinction, and (3) a coexistence state where the two strains coexist together. By using a quasi-steady state argument for $S$ we reduced the four dimensional system \eqref{eqn: model1} to a two dimensional system \eqref{eqn: trans_sys}. This simpler system exhibited a competitive exclusion equilibrium where the first strain forces the second strain to extinction and a coexistence equilibrium. Results and open questions are summarized in \autoref{tab: results}.

The model presented in this study uses a time delay to account for cross-immunity between two strains and is shown to be a harmless delay since it doesn't influence the stability of the boundary equilibrium points \citep{Gopalsamy_1983,gopalsamy_1984,Driver_72}. However, the time delay's influence on the stability of the coexistence equilibrium is an open question. This time delay acts as a definitive period for immunity as opposed to a continuous or distributed waning of protection \citep{Pell2022}. For comparison, we simulate the ODE version of the model \eqref{eqn: model1} with the $\beta_2R_lI_2$ terms replaced by $\epsilon\beta_2R_lI_2$ in order to study the effects of waning immunity, as shown in \autoref{fig:cross_imm}. As $\epsilon\to0$ (i.e. the waning period for cross-immunity increases) it is shown that the emergent strain requires more time to be established in the population if all parameters between the two strains are equal. Additionally, we can interpret the term $\varepsilon\beta_2I_2R_l$ as the number of new breakthrough infections that occur per time unit. As $\varepsilon$ increases to 1, the more likely breakthrough infections will occur. Using a similar model that does not account for demography, Boyle et al., showed that the rapid turnover from one variant to another is influenced by two components: the increase in transmissibility and the breakthrough infections \citep{boyle2022selective}. They deduce that emergent strains are the ones that are best at evading immunity \citep{boyle2022selective}. Our simulations in \autoref{fig:cross_imm} further support this.

\begin{figure}[!htbp]
    \centering
    \includegraphics[width=\textwidth]{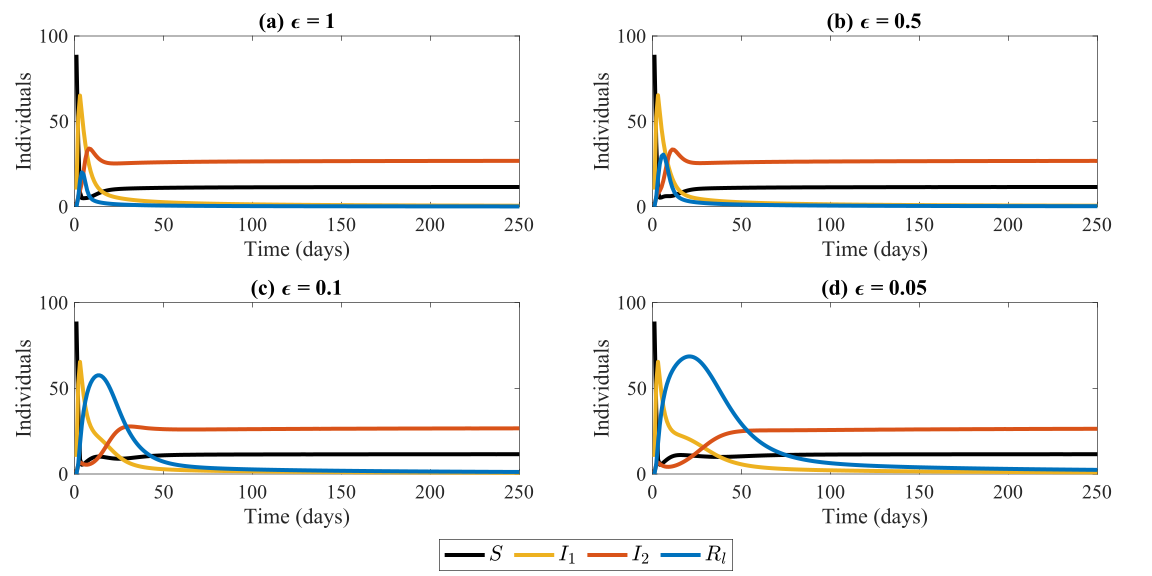}
    \caption{Simulations of the model \eqref{eqn: model1} with $\tau = 0$ where the $\beta_2R_lI_2$ term is replaced by $\epsilon\beta_2R_lI_2$ for the indicated values of $\epsilon$. Parameter values: $\beta_1 = \beta_2 = 0.03$, $\gamma_1 = \gamma_2 = 0.2$, $d_1 = d_2 = 0.15$, $a = 10$, and $d = 0.045$}
    \label{fig:cross_imm}
\end{figure}

We fit the model \eqref{eqn: model1} to wastewater data from the greater Boston area in order to show that the model can capture two-strain dynamics in the real world. Using wastewater data for fitting, as opposed to clinical data, allows us to avoid issues related to under-reporting or reporting lags of case data. We fit the model with and without time delay and found that incorporating time delay allowed the model to better follow the trend of the data for the first wave. However, regardless of the inclusion of time delay, the model did not qualitatively capture the second wave in the data (although the model with time delay performed slightly better). This is due to the models not accounting for the vaccination program that began around December. Although the issue of parameter identifiability is present (and beyond the scope of this study), we ultimately show that this four-dimensional model is able to capture complex two-strain dynamics. A local sensitivity analysis was carried out on the values obtained via curve-fitting and indicated that cumulative infections are sensitive to strain-specific contact and recovery rates. 

This paper may also be viewed as an extension of the work done by \cite{fudolig_2020}. While Fudolig and Howard did consider cross-immunity because they focused on SARS-CoV-2 and influenza strains co-circulating, we incorporated a time delay to account for one SARS-CoV-2 strain providing temporary immunity to another. Although that study included a compartment for vaccinated individuals, more direct comparisons may be made by setting the vaccination rate of their model ($p$) and the time delay of the model presented here ($\tau$) to zero. The authors derived the same local stability results for the disease-free and emergent strain (strain 2) equilibrium, and also found that the reproduction number of the emergent strain (strain 2) must be sufficiently small in order for the local stability of the established strain boundary equilibrium \citep{fudolig_2020}. Furthermore, we provide global stability results for the boundary equilibria. Our bifurcation plane for the full model, shown in \autoref{fig:R1R2 plane}, mirrors that of \cite{fudolig_2020}. In addition, we provide an analogous bifurcation plane for the transient system \eqref{eqn: trans_sys} in \autoref{fig:R02R02 plane}.

 SARS-CoV-2-infected individuals always go through a latent period, where they are yet to be transmissible clinically. This duration is related to the number of infectious viruses (or the within-host viral load) and should not be confused with the sub-clinical symptomatic phase, which can follow the latent period \citep{ke2021vivo, heitzman2022modeling}. Existing models examining SARS-CoV-2 transmission often consider latency, which better integrates epidemic data \citep{phan2023simple, patterson2022does, eikenberry2020mask}. However, for our analytical purposes, the inclusion of latency can complicate the mathematical analysis but usually has a small effect on the basic reproduction number and often does not affect global stability \citep{patterson2022does, van2017reproduction, feng2001role}. A similar simplification to facilitate model analysis was also done by \cite{boyle2022selective}. Thus, we made the simplifying assumption to not include latency in our current model.

 In general, immunity against one strain may not confer protection for a different strain, if the two strains are sufficiently different from one another. However, while the initial infection may be due to a single strain, mutations occur during the course of infection and may allow for the development of antibodies to various mutations of the initial strain, which may include the particular second strain. Yet more paradoxically, antibodies obtained from one strain may enhance the infection of another, which is known as the antibody-dependent enhancement of infection phenomenon \citep{junqueira2022fcgammar, maemura2021antibody, wan2020molecular, nikin2015role}. The evolutionary dynamic of SARS-CoV-2 itself is interesting and quite complex and should vary from individual to individual. Instead, the motivation for our model comes from the scenario when a mutant strain begins to emerge while another strain is dominant, as is the case of Alpha and Delta variants. In particular, taking into account the timing (e.g., the beginning of Delta vs. the end of Alpha) and scale differences in the number of infected individuals from each variant, we assume individuals recovered from Alpha can lose immunity and get infected with Delta during this time period. On the other hand, we assume individuals who recovered from Delta may not get infected with Alpha.
Due to this reason, we chose not to include vaccinations of individuals recovered from strain 2. In particular, if an individual is recovered from the emerging strain, regardless of the particular emerging variant, they may have some protection from the dominant strain. By the time the protection of this individual wanes, there should be much fewer individuals infected by the originally dominant strain to consider reinfection as a viable path of infection. Throughout the course of the SARS-CoV-2 pandemic, we have never observed a strain become dominant for multiple periods. Future work may consider extensions to these aspects of our model.

Ultimately, the model developed here, although simple in appearance, exhibits rich dynamics and, with the inclusion of wastewater-based epidemiology, is capable of capturing interactions of two strains circulating in the community. Future extensions of the model may include more than two strains and use standard incidence. For example, a model with $N$ strains may include $N$ infectious compartments but $N$ or fewer recovered compartments, depending on how cross-immunity is modeled. It may also be desirable to include a mutation factor to study the emergence mechanisms of various strains. Another fruitful direction would be to more realistically model the temporary cross-immunity period using a distributed delay framework.

\section*{Acknowledgements}
This work is supported by Faculty Startup funding from the Center of Infectious Diseases at UTHealth, the UT system Rising STARs award, and the Texas Epidemic Public Health Institute (TEPHI) to F.W. S.B. and Y.K. are partially supported by the US National Science Foundation Rules of Life program DEB-1930728 and the NIH grant 5R01GM131405-02. T.P. is supported by the director's postdoctoral fellowship at Los Alamos National Laboratory. We would like to thank the two anonymous reviewers for taking the time and effort necessary to review the manuscript. We sincerely appreciate all valuable comments and suggestions, which helped us to improve the quality of the manuscript.



\section*{Declarations}
\textbf{Competing Interests} The authors declare they have no competing interests.

\bibliography{sn-bibliography.bib}


\end{document}